\newif\ifabstract
\abstracttrue
 \abstractfalse 
\newif\iffull
\ifabstract \fullfalse \else \fulltrue \fi

\documentclass[11pt]{article}
\usepackage{amsmath}

\advance \oddsidemargin by -0.8in

\usepackage[T1]{fontenc}
\usepackage{color,colortbl}
\usepackage{cleveref}
\usepackage{comment}

\crefformat{section}{\S#2#1#3} 
\crefformat{subsection}{\S#2#1#3}
\crefformat{subsubsection}{\S#2#1#3}

 \usepackage{breqn}
\usepackage{hhline}

\usepackage{array}
\usepackage{pifont}
\usepackage{enumerate}

\usepackage{graphicx,epsfig,amsfonts,bbm,epstopdf,tabularx}

\usepackage{url}

\usepackage[utf8]{inputenc}
\usepackage[english]{babel}
 \usepackage{amsmath}
 \usepackage{amsthm}
 \usepackage{gensymb}

\newtheorem{theorem}{Theorem}[section]
\newtheorem{lemma}{Lemma}

\newcommand{\Exp}{\mathbb{E}}
\usepackage[ruled,linesnumbered,vlined]{algorithm2e}

\usepackage{stmaryrd}
\usepackage{bm}
\usepackage{multirow}

\usepackage{arydshln}
\usepackage{mathtools}
\usepackage{tabulary}
\usepackage{booktabs}
\usepackage{subfigure}
\usepackage{xspace}
\usepackage{appendix}

\usepackage{float}
\renewcommand*{\proofname}{Proof sketch}

\newtheorem{remark}{\textbf{Remark}}

\newcommand{\fixWideWidth}{0.49}
\newcommand{\fixWideHeight}{0.16}

\newcommand{\threesixty}{\protect 360\textdegree\xspace}
\newcommand{\algName}{\texttt{BOLA360}\xspace}
\newcommand{\tABR}{\texttt{ABR360}\xspace}
\newcommand{\AF}{\texttt{FOV}\xspace}

\newcommand{\VAts}{\texttt{VA-360}\xspace}
\newcommand{\PDash}{\texttt{ProbDASH}\xspace}
\newcommand{\SalVR}{\texttt{Salient-VR}\xspace}
\newcommand{\Flare}{\texttt{Flare}\xspace}
\newcommand{\Pano}{\texttt{Pano}\xspace}
\newcommand{\Mosaic}{\texttt{Mosaic}\xspace}

\newcommand{\hrPL}{\texttt{BOLA360-PL}\xspace}

\newcommand{\hrREP}{\texttt{BOLA360-REP}\xspace}

\newcommand{\addcite}[0]{\ifthenelse{\boolean{showcomments}}
	{\textcolor{blue}{~(add cite(s)) } }{}}
 
\newcolumntype{L}[1]{>{\raggedright\let\newline\\\arraybackslash\hspace{0pt}}m{#1}}
\newcolumntype{C}[1]{>{\centering\let\newline\\\arraybackslash\hspace{0pt}}m{#1}}
\newcolumntype{R}[1]{>{\raggedleft\let\newline\\\arraybackslash\hspace{0pt}}m{#1}}

\newboolean{showcomments}
\setboolean{showcomments}{true}

\definecolor{Green}{rgb}{0.0, 0.4, 0.0}

\newcommand{\revised}[1]{{#1}} 
\newcommand{\camRev}[1]{#1}
\newcommand{\TOMM}[1]{\textcolor{black}{#1}}

\AtBeginDocument{%
  }

\begin{document}

\title{BOLA360: Near-optimal View and Bitrate Adaptation for 360-degree Video Streaming}

\author{
Ali~Zeynali\thanks{University of Massachusetts Amherst. Email: {\tt azeynali@cs.umass.edu}.} \and
Mahsa~Sahebdel\thanks{University of Massachusetts Amherst. Email: {\tt msahebdelala@umass.edu}.} \and
Mohammad~Hajiesmaili\thanks{University of Massachusetts Amherst. Email: {\tt hajiesmaili@cs.umass.edu}.} \and
Ramesh~K.~Sitaraman\thanks{University of Massachusetts Amherst \& Akamai Technologies. Email: {\tt ramesh@cs.umass.edu}.} \and
}

\begin{titlepage}
\maketitle

\thispagestyle{empty}




\begin{abstract}
Recent advances in omnidirectional cameras and AR/VR headsets have spurred the adoption of 360\textdegree\xspace videos, which are widely believed to be the future of online video streaming. 360\textdegree\xspace videos allow users to wear a head-mounted display (HMD) and experience the video as if they are physically present in the scene. Streaming high-quality 360\textdegree\xspace videos at scale is an unsolved problem that is more challenging than traditional (2D) video delivery. The data rate required to stream 360\textdegree\xspace videos is an order of magnitude more than traditional videos. Further, the penalty for rebuffering events where the video freezes or displays a blank screen is more severe as it may cause cybersickness. We propose an online adaptive bitrate (ABR) algorithm for 360\textdegree\xspace videos called \algName that runs inside the client's video player and orchestrates the download of video \revised{tiles} from the server to maximize the quality-of-experience (QoE) of the user. \algName conserves bandwidth by downloading only those video \revised{tiles} that are likely to fall within the field-of-view (FOV) of the user. In addition, \algName continually adapts the bitrate of the downloaded video \revised{tiles} so as to enable a smooth playback without rebuffering. We prove that \algName is near-optimal with respect to an optimal offline algorithm that maximizes QoE. Further, we evaluate \algName on a wide range of network and user head movement profiles and \camRev{ show that it provides $6\%$ to $110\%$ improvements to the QoE of state-of-the-art algorithms.} While ABR algorithms for traditional (2D) videos have been well-studied over the last decade, our work is the first ABR algorithm for 360\textdegree\xspace videos with {\em both} theoretical and empirical guarantees on its performance.

\end{abstract}

\end{titlepage}


\maketitle

\section{Introduction}
\label{sec:introduction}



With recent advancements in omnidirectional cameras and AR/VR headsets, users can enjoy \threesixty media like YouTube~360 \cite{youtube360}, virtual and augmented reality 
applications \cite{psvr,gARVR}.
 Users either wear a head-mounted display (HMD) or use a device that allows them to change their viewport and field-of-view (\AF)\footnote{Field of view is the spatial area that falls within the viewport of the user's device. A user sees only the portion of the \threesixty video that is within the \AF.} when watching a \threesixty video (see Figure~\ref{fig:omni_dist}). For instance, a user watching World Cup soccer as a \threesixty video can wear an HMD and watch the game by changing their head position as if they were actually in the stadium.

The rapid increase in the popularity of \threesixty videos is partly driven by the wide availability of VR headsets that has grown more than five-fold in the past five years to reach nearly 100 million units in use~\cite{cicso_vni}.  A second trend driving the popularity of \threesixty videos is the wide availability of omnidirectional cameras that make it easy to create \threesixty video content.  While the promise of providing an immersive experience has made  \threesixty videos the holy grail of internet video streaming~\cite{ZinkSN19}, providing a high quality-of-experience to users while {\em delivering those videos at scale over the internet} is a major unsolved problem and is the main motivation of our work.

\begin{figure}[!t]
\centering
    \includegraphics[width=0.7\linewidth]{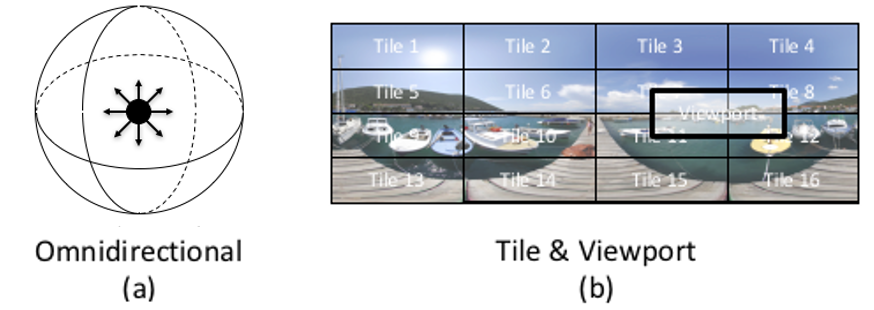}
   \caption{(a)  Users watch \threesixty videos by moving their viewport to point to any direction in the enclosing sphere (b) \revised{each frame of the \threesixty video is broken up into tile frames}~\cite{ZinkSN19}. }
   \label{fig:omni_dist}
\end{figure}

{\bf Tiled video delivery. } \revised{A common approach to deliver \threesixty video from server to users (i.e., client) is to divide the entire \threesixty video into same duration chunks of length $\delta$. Then, each chunk is spatially split into a set of tiles to fully cover the viewing sphere of the user (see Figure~\ref{fig:omni_dist}).} Each tile is encoded in multiple bitrates (i.e., resolutions) so that the quality of the tiles sent to the user can be adapted to the available bandwidth between the server and the client, a feature known as ``adaptive bitrate streaming''. Video tiles are streamed ahead of time and buffered at the client before they can be rendered to the user. As the user changes their viewport
the appropriate tiles within the user's  \AF  is extracted from the client's buffer and rendered on the user's display.


{\bf Challenges of \threesixty video delivery.} A key challenge in delivering \threesixty videos is that they are an order of magnitude larger in size than traditional (2D) videos~\cite{nguyen2018predictive, mangiante2017vr, dasari2020streaming}. \threesixty videos require multiple tiles to cover the entire viewing sphere, each encoded in multiple bitrates akin to 2D videos. Further, a high resolution of 4K to 8K is recommended for viewing AR/VR media \cite{mangiante2017vr}. Thus, the data rate of a \threesixty video that delivers a 4K stream for \camRev{eight tiles} and allows the user to watch the full \threesixty viewing sphere is 200 Mbps, compared to about 25 Mbps for a traditional 4K video. In fact, the data rate of such a \threesixty video is an order of magnitude larger than the US's average last-mile bandwidth \cite{usnetworkAkamai2017, usnetwork2021}. Additionally, when the user's viewport changes, say due to a head movement, the new tiles that fall within the user's new \AF must be rendered within a latency of a few tens of milliseconds so as to not cause a {\em rebuffering event} that results in showing either an incorrect/stale tile or no tile at all (i.e., blank screen). If the ``motion-to-photon'' latency exceeds a few tens of milliseconds, the user experiences a degraded quality-of-experience, or even cybersickness \cite{ZinkSN19}.

{\bf Adaptive Bitrate (ABR) for 360° Videos.} We investigate ABR algorithms for handling the challenges posed by the large size of 360° videos.
While ABR algorithms for traditional 2D videos have been extensively studied over the past decade \cite{kua2017survey, han2020vivo, mao2017neural, yin2015control, zhang2021sensei, spiteri2020bola, kim2020neural}, ABR algorithms for 360° videos are notably more complex. They must perform both "view adaptation" by predicting the user's head position and potential future tile views, and "bitrate adaptation" by determining appropriate bitrates for downloading tiles. Importantly, these two adaptations are jointly optimized to prioritize higher bitrates for tiles more likely to be in the user's viewport. Note that this challenge is different from tile scheduling problem which determines the download ordering of tiles~\cite{ghabashneh2023dragonfly}.

{\bf Challenges of Naive ABR Solutions.} Naive ABR algorithms equally distribute the available
bandwidth among all tiles,
resulting in downloading the same
 tiles for each chunk. \revised{While this approach prevents rebuffering by having the entire tiles of a chunk, it leads to suboptimal video quality. An alternative approach predicts the tiles the user is likely to watch and downloads only those tiles, reducing the number of downloaded tiles and allowing for higher quality.} However, this approach is susceptible to rebuffering if the user unexpectedly switches to unpredicted tiles
\cite{li2019very, fan2017fixation, xu2018gaze, xu2018predicting, ban2018cub360, xu2018tile, feng2020livedeep, feng2021liveroi}. Our proposed approach offers a provably near-optimal solution by striking a balance between these naive extremes, achieving both high quality and reduced rebuffering.



{\bf Our Contributions.}
We leverage Lyapunov optimization techniques to achieve {\em both} high bitrates and low rebuffering by judiciously downloading higher-quality tiles for tiles that are more likely to be in the \AF of the users, while using lower-quality tiles for the rest of the tiles as a hedge against rebuffering. Our algorithm, \algName, is a near-optimal ABR algorithm for \threesixty videos that also empirically performs better than state-of-the-art algorithms. We make the following specific contributions.

{\bf 1) } We frame the optimization of quality-of-experience (QoE) for \threesixty videos as the \tABR problem. We model QoE as a weighted sum of two terms, one term relates to the quality (i.e., bitrate) of the video tiles viewed by the user, and the other term relates to continuous video playback without rebuffers. 

\TOMM{\textbf{2)} We present an optimal offline solution to the \tABR problem, which establishes an upper bound on the achievable QoE by any online algorithm. While this offline algorithm is impractical for real-world use, it serves as a benchmark for comparing the QoE performance of online algorithms in our experimental analysis.}

\textbf{3)} We present \algName \footnote{\TOMM{A preliminary version of this article appeared at ACM MMSys 2024~\cite{zeynali2024bola360}.}}, an algorithm that finds a near-optimal solution for \tABR in an online manner without the future knowledge of inputs. In each round, \algName selects a suitable bitrate for each tile based on the current buffer utilization. Further, there are multiple parameters in \algName that could be tuned to improve the performance under different conditions and environments. 

\textbf{4)} We analyze \algName's performance, demonstrating that (i) it never exceeds the client's buffer capacity (Theorem~\ref{thm:buffer_size}) and (ii) its average QoE is within a small additive constant factor of the offline optimum of \tABR(Theorem~\ref{thm:alg_perf}), the additive factor goes to zero when the buffer size goes to infinity. Additionally, considering the {\em playback delay} – the time between tiled download and rendering, our analysis reveals a tradeoff between playback delay and \algName's QoE, i.e., one needs to tolerate a longer playback delay to achieve better QoE (Remark~\ref{rem:conflict}).


{\bf 5)} We implement \algName on a simulation testbed and evaluate its performance using  both real and synthetic data traces. Using trace-based simulations, we compare \algName with state-of-the-art algorithms used in \VAts~\cite{ozcinar2017viewport}, \PDash~\cite{xie2017360probdash},  \SalVR~\cite{wang2022salientvr}, \Flare~\cite{qian2018flare}, \Pano~\cite{guan2019pano}, and \Mosaic~\cite{park2021mosaic}.
Our results show that in comparison with QoE of the best alternative ABR algorithm, \camRev{on average \algName provides $6\%$ improvements over 14 real network profiles (Figure~\ref{fig:network_QoE_bitrate}) and $9\%$ improvements over} 12 different head position probability distributions (Figure~\ref{fig:prob_QoE_bitrate}).

{\bf 6)}  Finally, we explore two extensions to \algName, addressing specific real-world scenarios~\cite{spiteri2019theory}. While \algName already demonstrates impressive QoE, average bitrate, and rebuffering performance, further enhancements can be achieved by introducing heuristics on top of its core design.
We introduce 
\hrPL and \hrREP, each targeting specific limitations of the original algorithm. Our experiments show 
\hrPL reduces reaction time by up to $67.8\%,$ while \hrREP enhances both playing bitrate and reaction time by $91.2\%$ and $80.0\%,$ particularly when combined with short-term head position predictions. These heuristics offer efficient and practical solutions, surpassing the original algorithm's performance.


{\bf Roadmap.} The paper is structured as follows: \revised{First, we investigate the background of \threesixty video streaming in Section~\ref{sec:background}.} Then, we present the system model and formulate the \tABR problem in Section~\ref{sec:system_model}. Next, In Section~\ref{sec:offline}, we present an optimal offline solution for the problem of bitrate adaptation for \threesixty video streaming. In Section~\ref{sec:online}, we develop \algName using a Lyapunov optimization approach, proving its near-optimality. Section~\ref{sec:exp} evaluates the performance \algName against state-of-the-art algorithms. In Section~\ref{sec:heuristics}, we introduce two enhancements for \algName which practically improves its performance. The related work is discussed in Section~\ref{sec:relwork}, and we conclude in Section~\ref{sec:conc}.

\section{Background}
\label{sec:background}

{\bf ABR Algorithm for \threesixty Videos.} Tile-based \threesixty videos temporally slice the video into chunks. \revised{Each chunk is split into multiple tiles to cover the entire \threesixty spatial area. Usually, each tile is encoded in multiple quality levels or bitrates for video streaming.} The ABR algorithm for \threesixty video has to select the bitrate of a tile before downloading it. So, the action of the online ABR algorithm for each chunk is a list of selected bitrates for each tile.

{\bf Field of View [\texttt{FOV}].} A \threesixty video is encoded in the full \threesixty visual sphere. However, the human eye's field of vision  covers about 130\textdegree\cite{ratcliff2020thinvr}. Therefore, the user interacting with the \threesixty video cannot see the entire spatial area of the presented video. The part of the \threesixty video inside the user's visible region is called \texttt{Field of View} or \AF. Figure \ref{fig:background_FOV} shows an example of a
\AF that consists a subset of tiles of the full sphere of the \threesixty video seen by the user. We use the term \texttt{view} to refer to the group of tiles inside the \AF. When the user interacts with \threesixty video with a VR headset, the user can arbitrarily change the \AF and view by moving their head.

\begin{figure}
    \centering
    \includegraphics[width=0.5\linewidth, height=0.25 \linewidth]{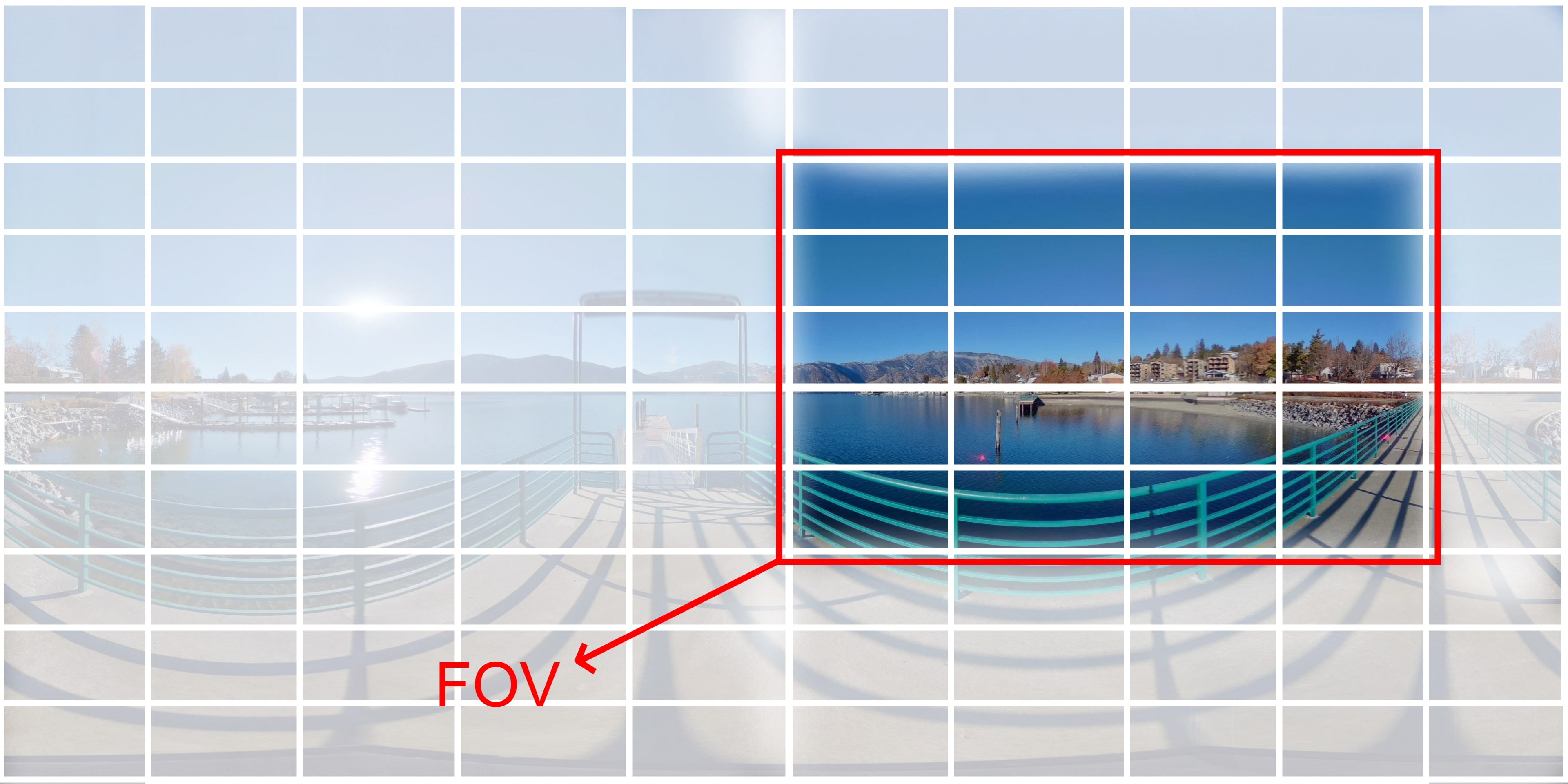}
    \caption{One shot from the entire spatial area of \threesixty video and \AF of user in that}
    \label{fig:background_FOV}
    \vspace{-3mm}
\end{figure}

\revised{{\bf Buffer Occupancy based Lyapunov Algorithm.} BOLA~\cite{spiteri2019theory} is an ABR algorithm optimized for single-tile 2D video streaming, using buffer occupancy in bitrate selection. In \threesixty video streaming, besides bandwidth uncertainty, additional factors like user head direction and \AF introduce complexity. Due to these added challenges and uncertainties unique to \tABR, traditional 2D ABR algorithms cannot be directly applied to effectively solve \tABR.}

\section{System Model and Problem Formulation}
\label{sec:system_model}
In this section, we present the system model and problem formulation for the online \threesixty video bitrate selection problem.

{\bf The \threesixty Video Model.} We consider a \threesixty video as a sequence of $K$ \textit{chunks}, where each chunk represents $\delta$ seconds of the playback time. \revised{Each chunk is further partitioned into $D$ \textit{tiles} to cover the entire \threesixty spatial area.} Each tile is encoded in $M$ different \textit{bitrates}, all of which are available at the server; the higher the bitrate, the larger the size in bits. Let $S_m$ denote the size \revised{(bits)} of a tile with bitrate \revised{index} $m$. We define $v_m$ as the utility value the user gets by watching a tile with bitrate index $m$. Therefore, we have the following inequality.
\begin{equation*}
   S_1 \leq S_2 \leq ... \leq S_M \Leftrightarrow   v_1 \leq v_2 \leq ... \leq v_M.
\end{equation*}
During the playback time of each chunk, the user views only tiles inside their \AF. The bitrate of tiles inside the \AF directly impacts the QoE. Downloading tiles which falls out of \AF wastes the bandwidth capacity. A key challenge is that the \AF is unknown to the bitrate selection algorithm at download time. As a result, the online bitrate selection algorithm must predict the \AF and download tiles based on its prediction. Let $p_{k, d}$ denote the probability of the tile $d$ is inside \AF while playing $k^{th}$ chunk. We assume that these probability values are given from a prediction based on the previous user's watching the video \cite{bao2016shooting, li2019very, chen2020sparkle, wang2022salientvr, park2019navigation}, or from a chunk analysis of the content \revised{combined with points probability analysis of \threesixty sphere \cite{xie2017360probdash,qian2016optimizing,xie2018cls}. For simplicity, we assume that the probability values of tiles within a chunk are normalized, such that $\sum_{d=1}^{D}p_{k,d} = 1$.}


\begin{table}[!t]
	\caption{\TOMM{Summary of important notations.}}
	\label{tbl:not_1}
	\begin{center}
		\begin{tabular}[P]{|c|L{9.2cm}|}
			\hline
			\textbf{Notation} & \textbf{Description} \\
			\hline \hline
			$K$ & Number of chunks\\

            $D$ & Number of tiles\\

            $M$ & Number of available bitrates\\

            $\delta$ & Length of a chunk\\

            $S_m$ & Size of a tile with bitrate index $m$\\

            $v_m$ & Utility value of watching a tile with bitrate index $m$\\

            $p_{k,d}$ & The probability of the tile $d$ is inside \AF while playing chunk $k$\\

            $T_{end}$ & Streaming duration\\

            $T_k$ & Time interval between finishing downloading chunks $k-1$ and $k$\\

            $\gamma$ & Relative importance of the two terms in user's QoE\\

            $Q(t)$ & Buffer level at time $t$\\

            $Q_{\max}$ & Buffer capacity\\

            $a_{k,d,m}$ & Decision variable for bitrate index $m$ of tile $d$ of chunk $k$\\

            $n_k$ & Average number of tiles downloaded for chunks played during the downloading of chunk $k$\\
			\hline 
		\end{tabular}
	\end{center}
\end{table}

\paragraph{Problem Formulation.}
In what follows, we formulate \tABR, an online optimization problem for the bitrate and view adaptation of \threesixty video streaming. In \tABR, the objective is to maximize the expected QoE of the user, including two terms: 1) the utility term that is related to quality of the video watched by the user, and 2) the smoothness of streaming term that captures continuous playback without rebuffering. The first term directly depends on the bitrate downloaded by the streaming algorithm, i.e., the higher the bitrate, the higher the utility. The second term captures the expected smoothness of video streaming. Rebuffering happens when at least one of the tiles inside \AF is not completely downloaded during playback time. Note that the above two terms conflict with each other. To maximize the utility, an ABR algorithm must download the highest possible bitrate tiles. However, to maximize the expected continuous smooth playback, the ABR algorithm must download low-bitrate tiles. Thus, to maximize the sum of both terms, the ABR algorithm must balance the two conflicting requirements.

We now formulate QoE mathematically to capture the utility as the sum of the two terms $U_K$ and $R_K$. The first term $U_K$ represents the time-average expected playback utility the video player prepares for the user over the sequence of \revised{chunks} and is defined as

\begin{equation}
\label{eq:UK}
    U_K =\frac{\sum_{k=1}^{K} \sum_{d=1}^{D} \sum_{m=1}^{M} \Exp \{  a_{k, d, m} \cdot p_{k, d} \cdot  v_m  \} }{\Exp \{ T_{\text{end}} \}},
\end{equation}

where $T_{\text{end}}$ is the time the video player finishes playback of the last chunk, and $a_{k, d, m}$ is a binary optimization variable in the \tABR problem: $a_{k, d, m}=1$ \revised{if bitrate index $m$ is selected for tile $d$ of chunk $k$; 0, otherwise.} The second QoE term is denoted by $R_K$, which targets the playback smoothness as follows.

\begin{equation}
\label{eq:RK}
    R_K =\frac{\sum_{k=1}^{K} \sum_{d=1}^{D} \sum_{m=1}^{M} \Exp \{  a_{k, d, m} \delta  \} }{\Exp \{ T_{\text{end}} \}}.
\end{equation}

That is, $R_K$ represents the ratio of the expected playback duration of downloaded tiles to the streaming duration. A low $R_k$ when $T_{\text{end}}$ greatly exceeds tiles playback duration (numerator) can lead to rebuffering, making a high $R_k$ indicative of continuous playback. Unlike $U_k$, $R_K$ inversely correlates with download time (or bitrate), decreasing with higher bitrates. Expectations in Equation~\eqref{eq:UK} and Equation~\eqref{eq:RK} are computed over the possible randomized decisions or outcomes of the ABR algorithm solving \tABR.

Let $t_k$ denotes the time the video player completes the download of tiles that belong to chunk $k-1$ and decides about the bitrate of tiles for $k^{th}$ chunk. And $T_k$ shows the time interval between finishing downloading chunks $k-1$ and $k$, i.e., $T_k = t_{k+1} - t_k$. We use the coefficient $\gamma > 0$ to set the relative importance of the two terms in the user's final QoE, i.e.,  $\gamma$ provides an opportunity to tune the relative importance of high-bitrate streaming with respect to a continuous streaming experience.
We formulate the \tABR problem as follows. 
\begin{subequations}
\label{eq:problem_form}
    \begin{eqnarray}
    \label{eq:problem_maximization}
      [\tABR] \qquad &\max &  U_K + \gamma R_K \\
      \label{eq:problem_1bitrate}
      &\textrm{s.t.,}& \sum_{m = 1}^{M} a_{k, d, m} \leq 1,\quad \forall d, k,\\
      \label{eq:problem_buffer}
       \quad && Q(t_k) \leq Q_{\max}, \; \quad \forall k, \\ 
       \label{eq:problem_decision_range}
      &\textrm{vars.,} & a_{k, d, m} \in \{0,1\}.
    \end{eqnarray}
\end{subequations}


Constraint \eqref{eq:problem_1bitrate} limits to \revised{select at most one bitrate for each tile of a chunk.} The second constraint \eqref{eq:problem_buffer} enforces the buffer capacity limit, where $Q(t_k)$ is the buffer level at time $t_k$ and \revised{ shows the aggregate length of tiles available in the buffer} at time $t_k$. $Q_{\max}$ is \revised{buffer capacity and depicts the maximum aggregate length of tiles stored in the buffer}. Since the number of tiles downloaded for each chunk is not fixed, the actual number of tiles that drain out from the buffer when a chunk is played can vary from chunk to chunk. To capture this, let $n_k$ be the average number of tiles downloaded for chunks played during the downloading of chunk $k$. The evolution of the buffer level is characterized as \revised{
\begin{equation}
\label{eq:buffer_update}
    Q(t_{k+1})  = \max  [Q(t_k) - n_k T_{k}, 0 ] + \sum_{d=1}^{D} \sum_{m = 1}^{M} a_{k, d, m}\delta,
\end{equation}
where the first term refers to the length of tiles removed from the buffer during the download time of chunk $k$ and the second term shows the length of tiles recently downloaded.}

\begin{remark}\label{rem:bandwidth2D}
    \TOMM{For regular 2D videos with $D = 1$, number of tiles that drain out of the buffer when each chunk is played fixed, $n_k = 1$. In this particular case, $\min  [Q(t_k), {T_{k}}]$ seconds drained out of the buffer after passing $T_k$ seconds.}
\end{remark}

\section{An Optimal Offline Algorithm}
\label{sec:offline}
\TOMM{In this section, we present an offline algorithm that obtains an upper bound for the optimal QoE of \tABR. The algorithm is listed as Algorithm~\ref{alg:opt_off} and is based on dynamic programming to optimally solve the \tABR problem, given the full knowledge of the bandwidth capacities in advance. While this algorithm is impractical since it requires the entire input a priori, we use this algorithm to evaluate the significance of the performance of the proposed online algorithm. 
It is worth noting that Algorithm~\ref{alg:opt_off} generates an upper bound for the offline solution of the \tABR problem formulated in Equation~\eqref{eq:problem_form}. However, it is not a fully offline algorithm since it still takes head position probability values as the input to \tABR.
Nevertheless, the performance of Algorithm~\ref{alg:opt_off} is an upper bound on the performance of any online algorithm without the knowledge of future bandwidth and head movement of the user. }

\TOMM{Now, we proceed to explain the details of Algorithm~\ref{alg:opt_off}. First, we discretize the time into slots with length $t_0$. Let $r(k, t, b)$ denote the maximum possible QoE the algorithm can achieve when it downloads the first $k$ chunks and finishes it at time $t$ and buffer level $b$. The algorithm initiates the value of $r(0, 0, 0) = 0$. Then, the algorithm calculates the download time $T$, and rebuffering time $R$, of any possible action that determines the selected bitrate during downloading $k^{th}$ chunk and evaluates the best possible performance by utilizing the values calculated for the first $({k-1})^{th}$ chunks. Let $T$ show how long the downloading of action $\textbf{a}$, the set of selected bitrates for each tile, would take. The algorithm may wait if the buffer is full to place the recently downloaded segments. Let $b'$ be the state of the buffer before downloading $k^{th}$ chunk, then the waiting time is at most $ T_0= \max [0, b' + n_k \delta - Q_{max}]$, where $n_k$ is the number of tiles with positive bitrate for $k^{th}$ chunk. Then, $T' = \lfloor T_0/t_0 \rfloor \times t_0$ is the conversion of $T_0$ to units of $t_0$. We round down the value of $T_0/t_0$ to ensure that the final calculated QoE is an upper bound for the QoE of the optimal offline solution. Also, we calculate the rebuffering happened while downloading $k^{th}$ chunk, $R = \max [T' - b', 0]$. Finally, we calculate the impact of action $\textbf{a}$ on achieved QoE by calculating $r'$ in line 16. Now, by knowing the download time and rebuffering for action $\textbf{a}$ we can update values of time, buffer, and utility achieved by the algorithm by the end of downloading $k^{th}$ chunk as demonstrated in Lines 11, 12, and 17 of Algorithm~\ref{alg:opt_off}. The following theorem states the optimality of Algorithm~\ref{alg:opt_off}. }

 \begin{figure}[t]
     \begin{algorithm}[H]
\SetAlgoLined
\KwResult{$\max_{(t,b)} r(K, t, b) / (t+b)$}
 $r(k,t,b) = -\infty$\;
 $r(0,0,0) = 0$\;
 \For{$k$ in $[1, 2, ...,K]$}{
 \For {all $(t', b')$ such that $r(k-1,t',b') > -\infty$}{
 \For {all possible set of action $\textbf{a} = [m^{(1)},m^{(2)},...,m^{(D)}]$}{
 $n:$ number of positive bitrates in $\textbf{a}$\;
 $T:$ download-time($\textbf{a}$)\;
 $T_0 = \max[ T, b'+ n \delta - Q_{max}]$\; 
 $T'$ = $\lfloor T_0/t_0 \rfloor \times t_0$\;
 $R = \max[T' - b', 0]$\; 
 $t = t' + T'$\;
 $b = b' - T' + R + n \delta$\;
 $\forall d, \ v^{(d)} := $ utility value of tile $d$\;
 $\forall d, \ a_{(d)} := 1$ if $m^{(d)} > 0$ else 0\;
 $r_k = \sum_{d=1}^{D} a_{(d)}(v^{(d)} p_{k, d} + \gamma \delta )$\;
 $r' = r(k-1, t', b') + r_k $\;
 $r(k, t, b) =  \max[r(k, t, b), r']$\;

 }
 }
 }
 \caption{Optimal offline algorithm for \tABR}
 \label{alg:opt_off}
\end{algorithm}
\vspace{-5pt}
 \end{figure}

\begin{theorem}
\label{thm:off_opt}
Algorithm \ref{alg:opt_off} gives an upper bound for QoE of optimal offline solution for \tABR.
\end{theorem}

\TOMM{{\em Proof.} The dynamic programming solution for the optimal offline problem with $K$ chunks evaluates the QoE that can be achieved by any permutation of actions during downloading the first ${K-1}$ chunks. Then, it uses the QoE values to evaluate the best achievable QoE for downloading all $K$ chunks. We prove the correctness of the algorithm by induction. The base case is evaluating the performance of the optimal offline solution for starting moment of the stream. $r(0,0,0) = 0$ is the base case which shows that the achieved QoE is zero at starting moment of the stream when we have downloaded no tiles and the buffer is empty. This proves the correctness of the base case. Since we are evaluating the performance of the offline algorithm for any value of $t$ and $b$ for $r(K-1, t, b)$,  we would be able to calculate the performance of the offline algorithm for all $K$ chunks by evaluating every possible action. It gives the optimal offline solution a chance to try any possible set of actions (line 5 of Algorithm \ref{alg:opt_off} ) and finds the best sequence of actions that leads to maximum QoE. As a result, $r(K, t, b)$ can store the maximum achieved QoE by downloading $K$ chunks until time $t$ and using $b$ seconds buffer. As a result, after filling the dynamic programming table, the value of $r(K, t, b) / (t+b)$ shows the upper bound on QoE of the optimal offline algorithm, and the proof ends.}

\section{\algName: An Online \threesixty ABR Algorithm }
\label{sec:online}

In this section, we propose \algName, a Lyapunov-based algorithm that finds a near-optimal solution to \tABR. \algName is an online algorithm whose decisions do not require the knowledge of future bandwidth values.

\subsection{Design and Analysis of \algName}
The design of \algName is based on three key ideas. First, \algName finds a solution for a single-slot maximization problem that leads to a near-optimal solution for the original long-term problem over $K$ chunks. Note that, solving the long-term optimization problem is not possible for the online algorithm since there is uncertainty about the future input. Second, the single-slot decision of \algName is based on the buffer level; the higher the current buffer level, the higher the selected bitrate for download. This is intuitive since a high buffer level indicates that the input rate into the buffer was higher than the output rate from the buffer, so the algorithm has more freedom to download high-quality tiles. Third, \algName uses a threshold as the indicator of high buffer utilization, and by reaching the threshold, it moves to an idle state and waits until the buffer level decreases again. This approach limits the buffer utilization of \algName. It is worth noting that at the beginning and with an empty buffer, \algName starts downloading low bitrates.
With the above three key ideas, we now explain the technical details of \algName. The pseudocode for action taken by \algName for chunk $k$ is described in Algorithm \ref{alg:alg_pseu}.

\algName uses an input parameter $V$ that controls the trade-off between the performance of the algorithm and the maximum acceptable buffer utilization of the algorithm. Note that parameter $V$ also plays a critical role in the playback delay, i.e., for real-time streaming, smaller values of $V$ are preferable, while in an on-demand streaming application, the larger values of $V$ are acceptable. At the decision time $t_k$ for chunk $k$, the buffer level $Q(t_k)$ and head position probability values encoded in $p_{k, d}$ are given. \algName selects the bitrates for tiles of chunk $k$ by solving the maximization problem described in the following. 

\begin{subequations}
\label{eq:alg_act}
    \begin{align}
    \label{eq:alg_maximization}
    \text{arg}\max_{a(k)}&\quad  \eta(k, a(k)) = \sum_{d=1}^D \sum_{m=1}^M \frac{  a_{k, d, m} \big(V (v_m \ . \ p_{k, d}  + \gamma \delta) - Q(t_k)/\delta\ \big) }{ S_m}\\
       \label{eq:alg_single_bitrate}
      &\quad \textrm{s.t.,}\quad \quad  \sum_{m = 1}^{M} a_{k, d, m} \leq 1,\quad \forall k, d,\\ 
      &\quad \textrm{vars.,}\quad \  a_{k, d, m} \in \{0,1\},
    \end{align}
\end{subequations}

\TOMM{where 
$$a(k) :=\{a_{k,d,m} | \forall k, m \},$$}
is a decision vector of \algName and
$$ 0 < V \leq \frac{Q_{\max}/\delta -D}{v_M+\gamma \delta},$$
is a control parameter bounded by the R.H.S term to guarantee that the required buffer level for \algName is less than $Q_{\max}$.
\revised{Constraint \eqref{eq:alg_single_bitrate} limits \algName to select at most one bitrate for each tile.} \algName selects the near-optimal bitrates of chunk $k$ by finding a decision vector $\textbf{a}(k) = [a_{k, 1, 1}, a_{k, 1, 2}, ..., a_{k, 1, M}, a_{k, 2, 1}, ..., a_{k, D, M}]$ 
that maximizes the value of $\eta(k, a(k))$ in Equation~\eqref{eq:alg_maximization}. When the buffer level exceeds $V\delta (v_M  + \gamma \delta)$, the algorithm enters the idle state and downloads nothing. In this situation, \algName waits for $\Delta$ seconds and repeats the bitrate selection for that chunk again. The selection of $\Delta$ could be dynamic as suggested in~\cite{spiteri2020bola}, the algorithm waits until the buffer level reaches $Q(t_0) \leq V\delta(v_M + \gamma \delta )$. We note that our theoretical analysis is valid even with a dynamic waiting time. 

\begin{algorithm}[!h]
\SetAlgoLined
\TOMM{$\textbf{a}(k)$: A decision vector that maximizes the value of $\eta(k, a(k))$ defined in \eqref{eq:alg_maximization} with respect to single-bitrate constraint \eqref{eq:alg_single_bitrate} for chunk $k$\;}

 \If {number of non-zero elements in $\textbf{a}(k) >$ 0}
 {
 Download bitrates according to $\textbf{a}(k)$ and finish the decision making of chunk $k$\;
 }
 \Else{
 Wait for $\Delta$ seconds and repeat the bitrate selection for this chunk again\;
 }
 \caption{\algName($k$)}
 \label{alg:alg_pseu}
\end{algorithm}


\subsection{Theoretical Analysis of \algName}
We first provide an upper bound for the buffer level of \algName in Theorem~\ref{thm:buffer_size}. Second, in Theorem~\ref{thm:alg_perf}, we show the QoE of \algName is within a constant term of the optimal QoE of \tABR. The theoretical results reveal an interesting trade-off between the QoE and the playback delay of the \algName, which is discussed in Remark~\ref{rem:conflict}.

\begin{theorem}\label{thm:buffer_size}
Under bitrate control of \algName, the buffer level never exceeds $V\delta (v_M + \gamma \delta ) + D\delta$.
\end{theorem}

{\em Proof.} The proof of this theorem is inspired by the proof of Theorem~1 in \cite{spiteri2020bola}. However, for \algName, one has to deal with another challenge originated by adding head position probabilities into the control plane of \algName. The high-level idea is \algName select bitrates if the buffer level is at most $V\delta( v_M + \gamma \delta)$, otherwise it enters the idle states. Therefore, the maximum possible value of buffer level after download of new tiles would be $V\delta( v_M + \gamma \delta) + D\delta$. \TOMM{We prove this theorem by induction. The base case is $Q(t_1) = 0 \leq V \delta (v_M + \gamma \delta ) + D\delta $ that satisfies the statement of the theorem. Two cases are possible for value of $Q(t_k)$:
\begin{itemize}
    \item \textit{Case 1}: $Q(t_k) \leq V \delta (v_M + \gamma \delta )$: in this case, the buffer level at time $t_{k+1}$ will not exceed $Q(t_k) + D \delta \leq V \delta (v_M + \gamma \delta ) + D\delta .$
    \item \textit{Case 2}: $V \delta (v_M + \gamma \delta ) < Q(t_k) \leq V \delta (v_M + \gamma \delta  ) + D\delta $: In this scenario, the action at time $t_k$ is to wait and refrain from downloading any tiles, as downloading any bitrate $m$ would introduce a negative term into the value of $\eta(k, \textbf{a}(k))$ in Equation~\eqref{eq:alg_maximization}. Thus, $Q(t_{k+1}) \leq Q(t_k)$.
\end{itemize}}

Now, we proceed to analyze the QoE of \algName. With large $K$, the \tABR problem with rate stability constraint \cite{neely2010stochastic} is equivalent to the relaxed version of \tABR with limited buffer capacity, i.e.,

\begin{subequations}
\label{eq:rate-stability}
\begin{align*}
      Q(t_k) \leq &  Q_{\max}\\ 
      \Rightarrow &\lim_{K \to \infty} \frac{1}{K} \Exp \bigg\{\sum_{k = 1}^{K} \sum_{d=1}^{D} \sum_{m=1}^{M} a_{k, d, m} \delta \bigg\} \leq \lim_{K \to \infty} \frac{1}{K} \Exp \bigg\{\sum_{k = 1}^{K} n_k T_k \bigg\}.
\end{align*}
\end{subequations}
In addressing the \tABR problem with limited buffer capacity, it's essential to ensure that the expected input rate into the buffer remains below the buffer's output rate. Failing to do so could lead to a buffer capacity breach, especially as $K$ approaches infinity. Notably, solutions that accommodate limited buffer capacity inherently satisfy the rate stability constraint, though the reverse may not always hold. \revised{In addition, in the limited buffer capacity setting, the difference between $\Exp\{T_{\text{end}}\}$ and $\Exp\{\sum_{k=1}^K T_k\}$ is bounded by a finite value of $Q_{\max}$. Consequently, for the large videos, Equations~\eqref{eq:UK} and~\eqref{eq:RK} allow the substitution of $\Exp\{T_{\text{end}}\}$ with $\Exp\{\sum_{k=1}^K T_k\}$.}


\textbf{\textit{The stationary algorithm}.} In the context of \tABR problem, we define \textit{stationary algorithm} as an ABR algorithm that uses a fixed set of bitrates, $\textbf{A}^{*}$, with size $D$ ($|\textbf{A}^{*} | = D$), and for each chunk $k$, the set of selected bitrates for all $D$ tiles are the same as the $\textbf{A}^{*}$. Note that the selected bitrate for each tile may vary over time depending on the head position probability values, while the set of bitrates selected for all tiles of the chunk remains fixed.

Offline \tABR problem fits in the notation of optimization for renewal frames \cite{neely2012dynamic}. Precisely, by setting renewal frame duration the same as chunk download times and letting the achieved QoE of downloading each chunk represent penalty values in the notation of \cite{neely2012dynamic}, the offline \tABR problem can convert into an optimization problem over renewal frames. Then, following Lemma 1 in \cite{neely2012dynamic}, we prove the existence of a stationary algorithm with optimal QoE of $U_K^{*} + \gamma R_K^{*}$.

\begin{lemma}
\label{lem:stat_opt}
For the \tABR with a large video, i.e., $K~\to~\infty$, there exists a stationary algorithm that satisfies the rate stability constraint and achieves the optimal expected QoE of $U_K^{*} + \gamma R_K^{*}$.
\end{lemma}
{\em Proof Sketch.} The proof of this lemma follows from Lemma 1 in \cite{neely2012dynamic} and continues with the approach taken for proof of Lemma 1 in \cite{spiteri2020bola}. Based on the definition of a stationary algorithm for the \tABR problem, the expected QoE of the stationary algorithm is the same as expected, achieving QoE on each slot, which satisfies the criteria of Lemma 1 in \cite{neely2012dynamic}.

\begin{theorem}[main theorem]
\label{thm:alg_perf}
Let  $\texttt{OBJ}$ be the  expected QoE achieved by \algName. For a large video, i.e., $K \to \infty$, 
\begin{equation}
    \label{eq:performance}
   \texttt{OBJ}^* - \frac{D \delta^2 + \Psi}{2V\delta^2}\sigma \leq \texttt{OBJ} ,
\end{equation}

where  $\texttt{OBJ}^*= U_K^* + \gamma R_K^*$ is expected QoE of the offline optimal algorithm, and $\sigma = 1/\Exp \{ T_{k}\}$ and $\Psi \leq \Exp\{ D T_k^2\}$. That is, \algName achieves a QoE that is within an additive factor of the offline optimal. 
\end{theorem}

{\em Proof.} Let's define the Lyapunov function $L(Q(t_k))$, and per-slot conditional Lyapunov drift $\Phi(t_k)$ as below
\begin{align*}
L(Q(t_k)) &= \frac{1}{2\delta^2} Q^2(t_k),\\
\Phi(t_k) &= \Exp \{ \Delta\ L(Q(t_k))\  |\ Q(t_k) \} = \Exp \{L(Q(t_{k+1})) - L(Q(t_k))  | Q(t_k) \}.
\end{align*}


\TOMM{
Consider two cases: 1) $Q(t_k) \leq n_k T_k$ and, 2) $Q(t_k) > n_k T_k$. The value of $\Phi(t_k)$ can be derived from the buffer level evolution described in Equation \eqref{eq:buffer_update}.In the first case, when $Q(t_k) \leq n_k T_k$, the buffer is emptied after downloading chunk $k$, and the value of $\Phi(t_k)$ is given by:}
\begin{align*}
    \Phi_1(t_k) = \Exp \{ \frac{1}{2} \big( \sum_{d=1}^D\sum_{m=1}^M a_{k, d, m} \big)^2  - \frac{1}{2\delta^2} Q^2(t_k) | Q(t_k)\},
\end{align*}
 \TOMM{In the second case, where $Q(t_k) > n_k T_k$, we have:}
\begin{align*}
    \Phi_2(t_k) = \Exp \big\{\frac{1}{2} \big(\sum_{d=1}^D\sum_{m=1}^M a_{k, d, m} - \frac{n_k T_k}{\delta} \big) ^2  -\frac{Q(t_k)}{\delta} \big( \frac{ n_k T_k}{\delta} - \sum_{d=1}^D\sum_{m=1}^M a_{k, d, m} \big)\big\} | Q(t_k) \bigg\}.
\end{align*}

\TOMM{Thus, the value of $\Phi(t_k)$ can be expressed as:}
\begin{subequations}
\begin{align*}
    \Phi(t_k) & \leq  \max  \{  \Phi_1(t_k),  \Phi_2(t_k) \}\\
    & \leq  \Exp \bigg\{ \frac{n^2_k T_k^2 + \delta^2 (\sum_{d=1}^D\sum_{m=1}^M a_{k, d, m})^2}{2\delta^2} | Q(t_k) \bigg\}-  Q(t_k) \Exp\bigg\{ \frac{n_k T_k}{\delta}  - \sum_{d=1}^D\sum_{m=1}^M a_{k, d, m} \} | Q(t_k) \bigg\}\\
    & \leq \Exp\bigg\{ \frac{D \delta^2 + n_k^2 T_k^2}{2\delta^2} | Q(t_k) \bigg\} - Q(t_k)\Exp\bigg\{\frac{n_k T_k }{\delta} -\sum_{d=1}^D\sum_{m=1}^M a_{k, d, m} \} | Q(t_k) \bigg\}\\
    & =  \frac{D \delta^2 + \Psi}{2\delta^2} - Q(t_k)\Exp\bigg\{\frac{n_k T_k}{\delta} -\sum_{d=1}^D\sum_{m=1}^M a_{k, d, m} \} | Q(t_k) \bigg\}.
\end{align*}
\end{subequations}

\TOMM{By subtracting a fixed term from both sides we get:}
\begin{subequations}
\begin{align*}
    \Rightarrow \Phi(t_k) - & V\Exp\bigg\{ \sum_{d=1}^D  \sum_{m=1}^M a_{k, d, m}(p_{k,d} v_m+\gamma \delta)| Q(t_k) \bigg\} &\\
    \leq &   \frac{D \delta^2 + \Psi}{2\delta^2} - \frac{Q(t_k)}{\delta} \Exp\bigg\{ \frac{n_k T_k}{\delta} - \sum_{d=1}^D\sum_{m=1}^M a_{k, d, m}| Q(t_k) \bigg\}\\
    - & V\Exp\bigg\{ \sum_{d=1}^D  \sum_{m=1}^M a_{k, d, m}(p_{k,d} v_m+\gamma \delta)| Q(t_k) \bigg\}\\
    \leq &   \frac{D \delta^2 + \Psi}{2\delta^2} - \frac{Q(t_k)}{\delta} \Exp\bigg\{ \frac{n_k T_k}{\delta} - \sum_{d=1}^D\sum_{m=1}^M a_{k, d, m}| Q(t_k) \bigg\}\\
    -& V(U_K^* + \gamma R_K^*) \Exp \{ T_{k} | Q(t_k)\}.
\end{align*}
\end{subequations}

\TOMM{The previous equation holds since the decision of \algName at time $t_k$ is a solution of the maximization equation detailed in Equation \eqref{eq:alg_act} and}
$$\Exp \big\{ \sum_{d=1}^D \sum_{m=1}^M a_{k, d, m}^{*}(p_{k,d} v_m+\gamma \delta)|Q(t_k)\big\} \leq \Exp \big\{ \sum_{d=1}^D \sum_{m=1}^M a_{k, d, m}(p_{k,d} v_m+\gamma \delta)|Q(t_k)\big\}.$$

Then, we have
\begin{subequations}
\begin{align*}
     \Phi(t_k) - & V\Exp\bigg\{ \sum_{d=1}^D  \sum_{m=1}^M a_{k, d, m}(p_{k,d} v_m+\gamma \delta)| Q(t_k) \bigg\} \\
    \leq & \frac{D \delta^2 + \Psi}{2\delta^2} - \frac{Q(t_k)}{\delta} \bigg( \frac{\Exp{\{n_k\}}}{\delta} - \frac{\Exp \{ \sum_{d=1}^D\sum_{m=1}^M a^{*}_{k, d, m} \}}{\Exp\{T^{*}_k\}} \bigg)\Exp\{T_k \}\\
    - & V(U_K^* + \gamma R_K^*) \Exp \{ T_{k}\},
\end{align*}
\end{subequations}
where $a^{*}_{k, d, m}$ is the action of stationary algorithm for tile $d$ and bitrate index $m$ of chunk $k$ which satisfies the rate stability constraint. $T^{*}_k$ shows the length of download time for chunk $k$ while the stationary algorithm is taking action. Based on rate stability constraint, the second term in the equation above is always negative,
\begin{subequations}
\begin{align*}
    \sum_{k=1}^K \Phi(t_k) & - \sum_{k=1}^K V\Exp\bigg\{ \sum_{d=1}^D  \sum_{m=1}^M a_{k, d, m}(p_{k,d} v_m+\gamma \delta )| Q(t_k) \bigg\}\\
    \leq & \frac{D \delta^2 + \Psi}{2\delta^2}K  -  V(U_K^* + \gamma R_K^*) \Exp \{ T_{k}\} K.
\end{align*}
\end{subequations}

By dividing all terms by $V \cdot K \cdot  \Exp\{T_{k} \}$ and taking the limit $K \to +\infty$, the proof is completed, noting that the total download time is at most $Q_{\max}$ seconds shorter than $T_{end}$.


\begin{remark}[On the conflict between the playback delay and QoE of streaming]
\label{rem:conflict}
Theorem~\ref{thm:alg_perf} states as the value of $V$ increases, the performance of \algName gets closer to the optimal QoE. However, Theorem \ref{thm:buffer_size} reveals that the upper bound on the playback delay increases with higher values of $V$. Comparing these results, we observe a trade-off between minimizing playback delay and maximizing QoE in \algName. As the playback delay increases, the QoE performance of \algName approaches the offline optimum.

\end{remark}





\subsection{Understanding the  Behavior of \algName}
\label{sec:ex_simple}
We demonstrate the functionality of \algName through a straightforward test. Our test utilizes a 250-second video, segmented into 5-second chunks. Each chunk further divides into six tiles, each encoded at distinct bitrates: \revised{2Mbps, 4Mbps, 6Mbps, 8Mbps, 10Mbps, and 15Mbps}. To represent utility values, we employed a logarithmic function $v_m = \log(2S_m / S_1)$, similar to previous works such as \cite{reichl2013logarithmic, spiteri2020bola, hu2019optimization}. While our theoretical results only require a non-decreasing utility function, we opted for a concave function that better reflects real-world utility functions. The concave utility function exhibits a diminishing return property, meaning that increasing the bitrate from 1 Mbps to 2 Mbps provides more utility than increasing it from 10 Mbps to 11 Mbps, even though the bitrate difference is the same in both cases. In Table~\ref{tbl:ex1_bitrates}, we report the utility values generated from logarithmic function, size of tiles, and available bitrates. For this simple test, we set $\gamma = 0.1$ and $V = 5.5$. \revised{We note that $U_K$ assesses the expected utility across various tiles within a chunk, while $R_K$ quantifies the aggregate length of downloaded tiles. Setting $\gamma = 1/D$ equalizes the significance of utility and smoothness concerning a single tile.}

\begin{table*}[!t]
	\caption{\TOMM{Available bitrates and utility values used in Section~\ref{sec:ex_simple}}}
	\label{tbl:ex1_bitrates}
	\begin{center}
		\begin{tabular}[P]{|L{3cm}|c|c|c|c|c|c|}
			\hline
			\textbf{Bitrate (Mbps)} & 2 & 4 & 6 & 8 & 10 & 15 \\
			\hline
			\textbf{Sizes (Mb)} & 10 & 20 & 30 & 40 & 50 & 75 \\
			\hline
			\textbf{Utility values} & 1.000 & 2.000 & 2.585 & 3.000 & 3.322 & 3.907 \\
			\hline

		\end{tabular}
	\end{center}
\end{table*}


\begin{figure*}
	\centering
    \subfigure{\includegraphics[width=0.4\textwidth]{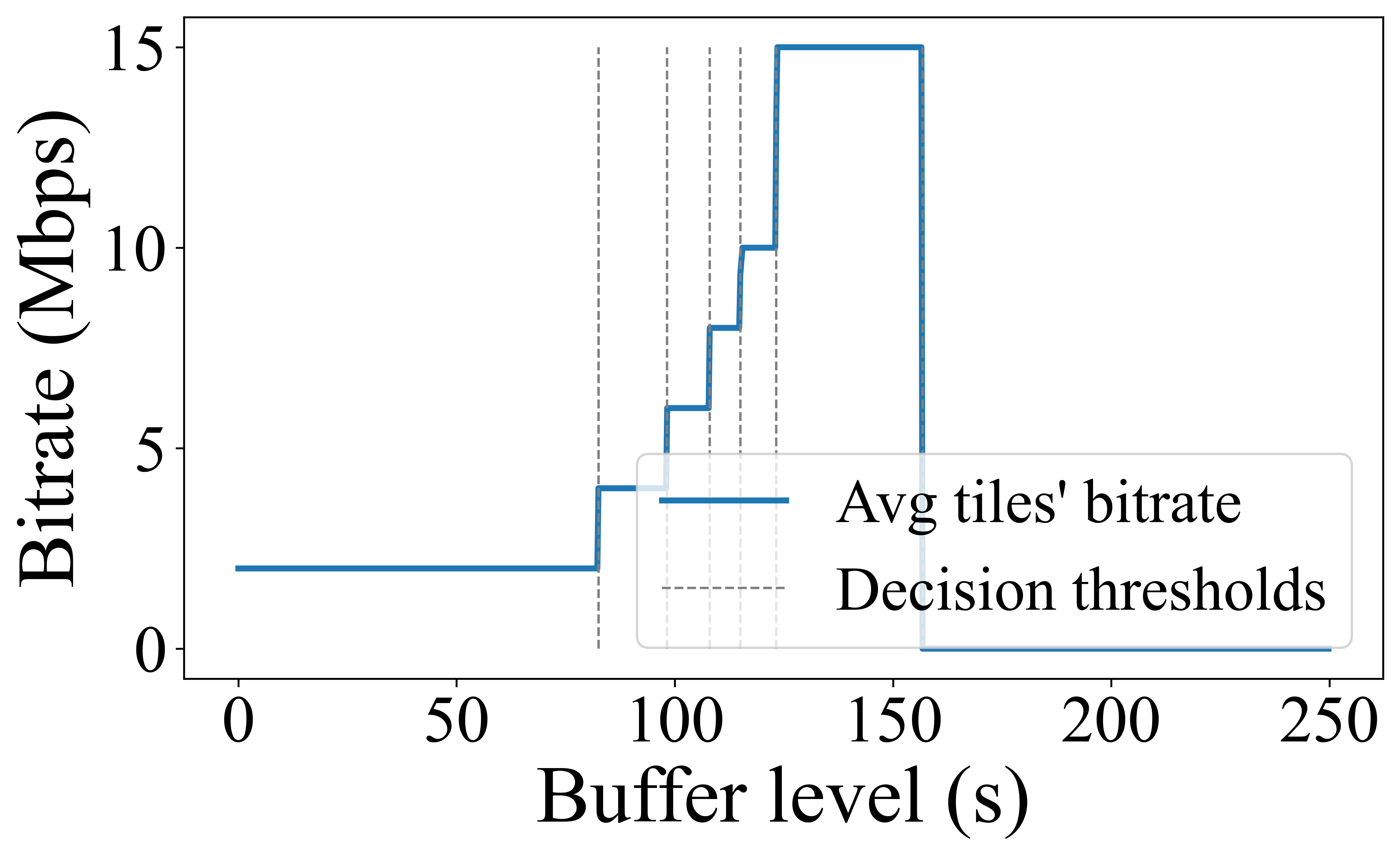}}\hspace{8mm}
	\subfigure{\includegraphics[width=0.4\textwidth]{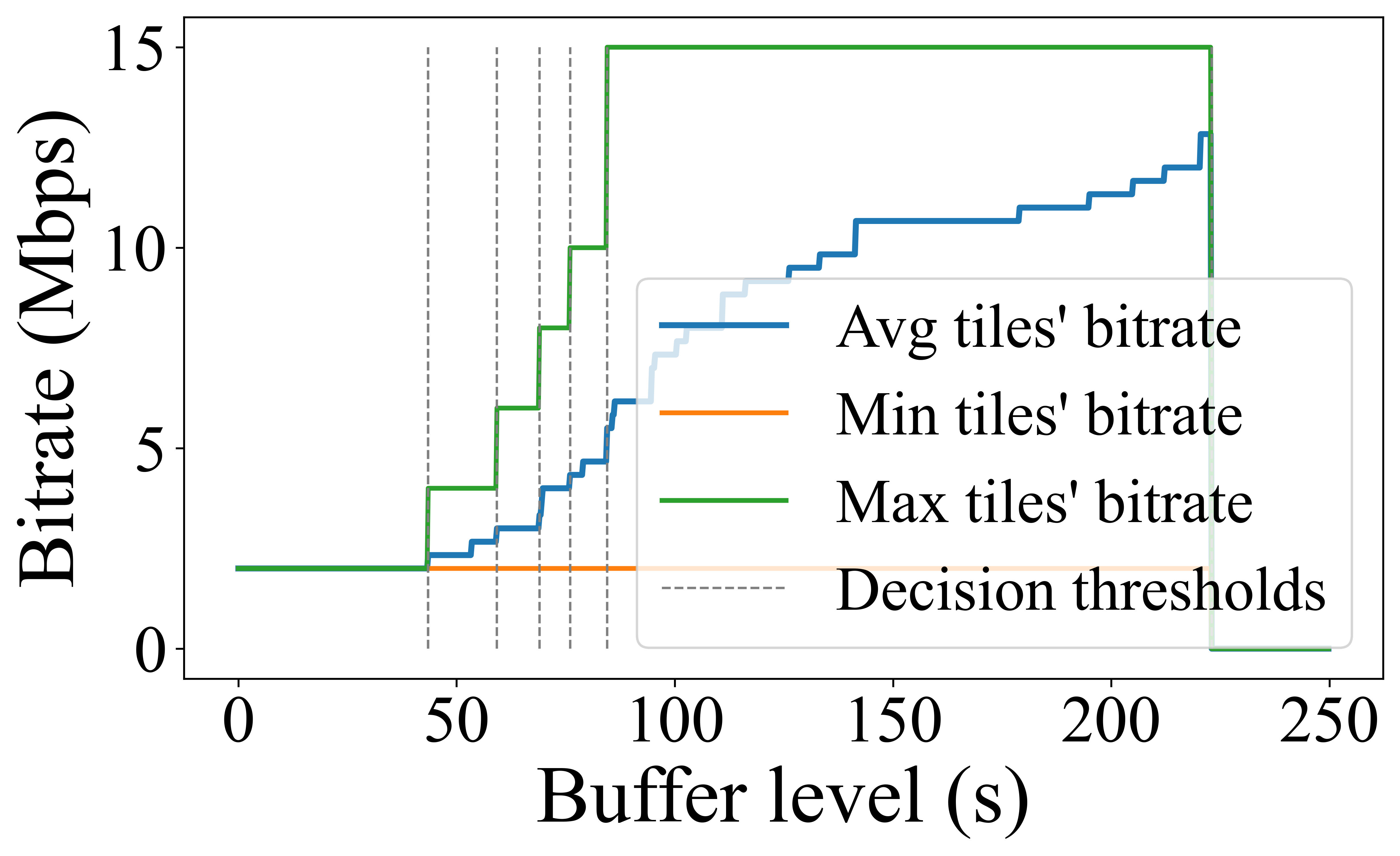}}
	\caption{The selected bitrate of \algName for tiles with highest and lowest probability and average selected bitrate as a function of buffer level for homogeneous (left) and heterogeneous (right) distributions.}
	\label{fig:mic_rates2}
\end{figure*}

\begin{figure*}
	\centering
	\subfigure{\includegraphics[width=0.4\textwidth]{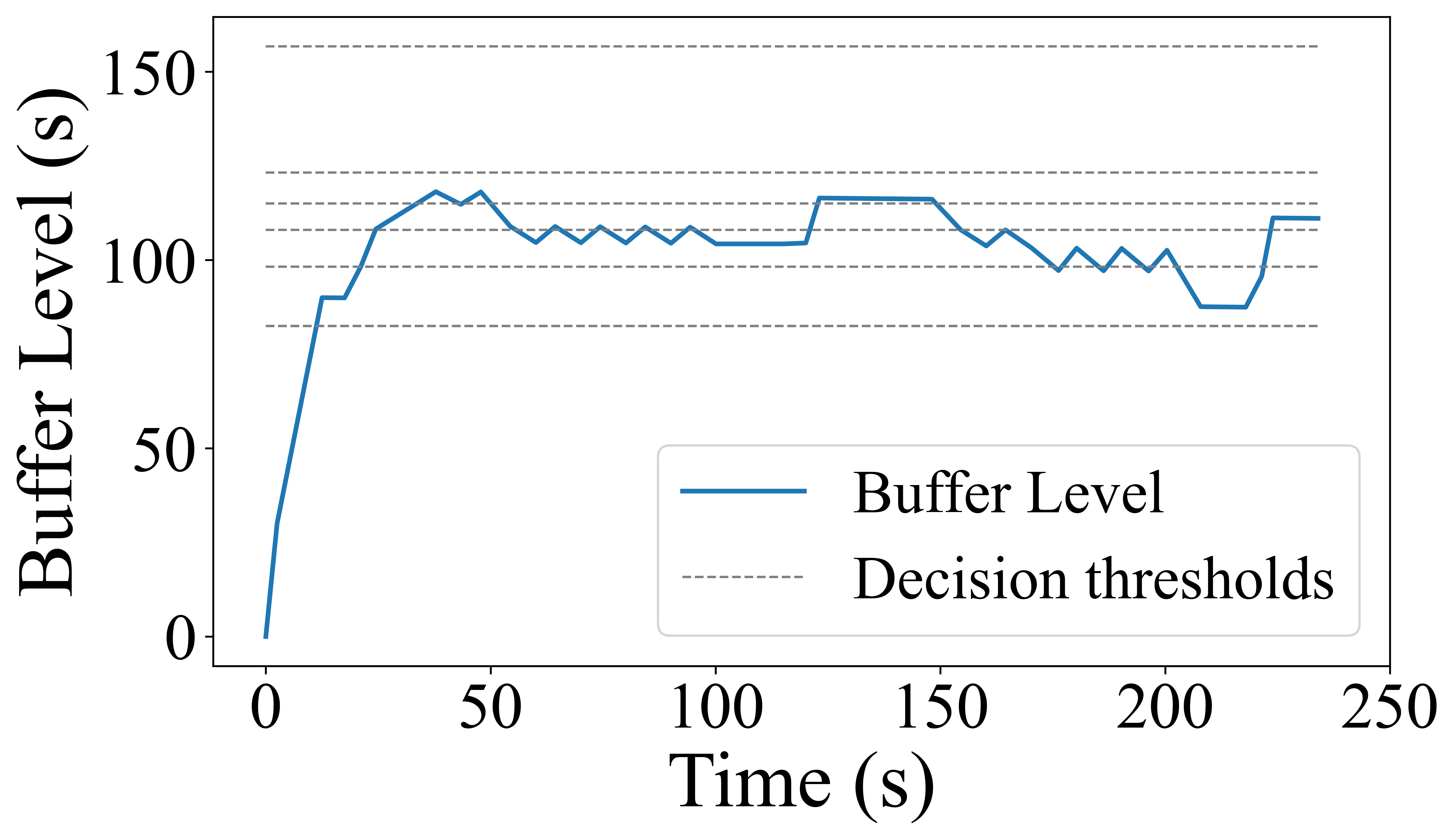}}\hspace{8mm}
	\subfigure{\includegraphics[width=0.4\textwidth]{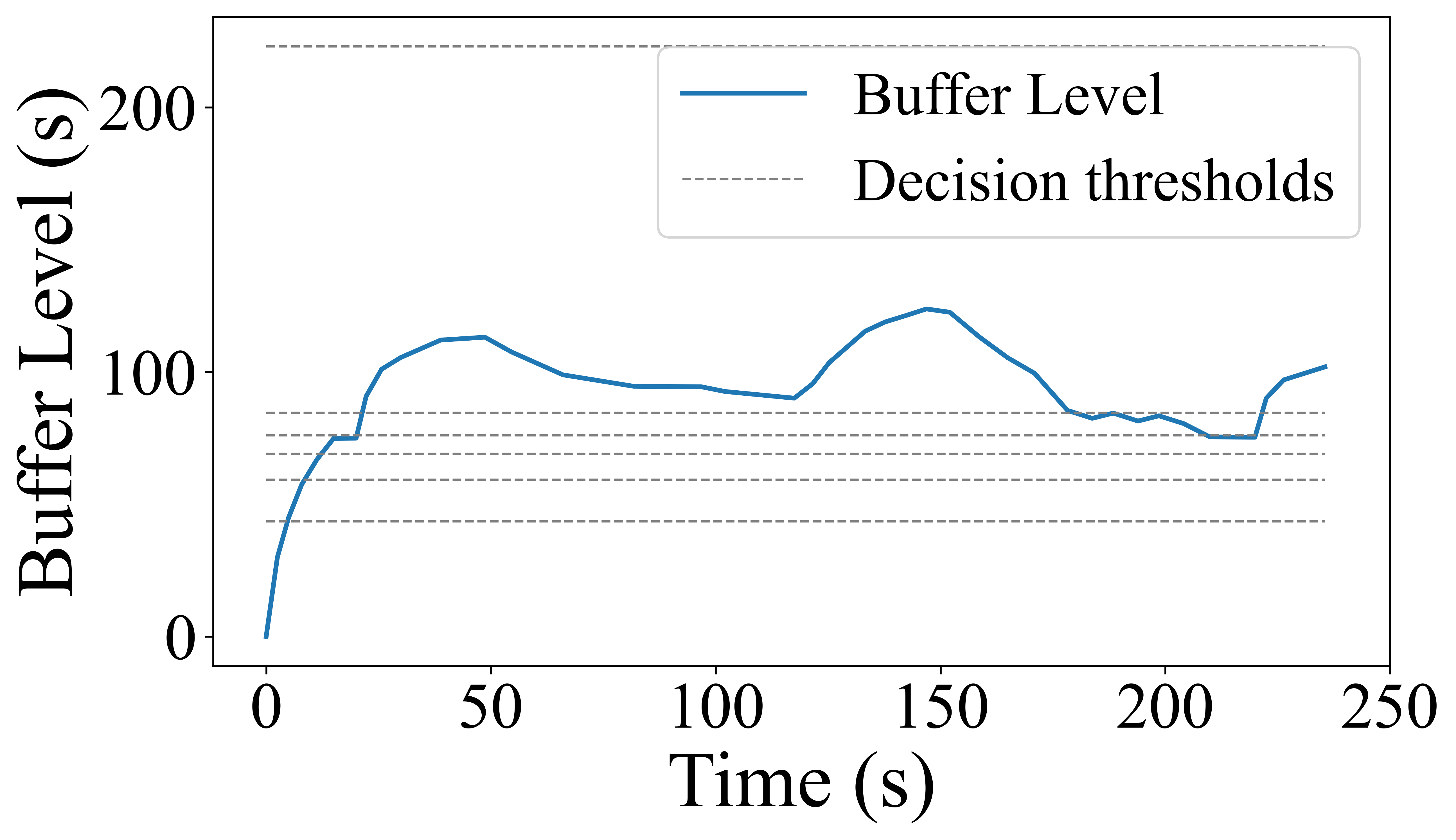}}
	\caption{Buffer level variation over time under bitrate selection of \algName for homogeneous (left) and heterogeneous (right) distributions.}
	\label{fig:mic_rates}
\end{figure*}

The head position of the user 
is represented by a 
probability distribution that is critical for guiding the actions of \algName. For this test, we evaluate the performance of \algName using two different head position probability distributions. The first distribution is homogeneous, where each tile is assigned a uniform probability, resulting in an equal likelihood of the user watching any tile ($p_{k, d} = 1 / D$ for all tiles). The second distribution is heterogeneous, with a linear increase in probability from the minimum to the maximum. Specifically, we set the maximum and minimum probabilities as $0.317$ and $0.017$, respectively. \TOMM{The values of $p_{k, d}$ for these two head position probability distributions are listed in Table \ref{tbl:ex1_probs}.} \revised{Note that these values are chosen arbitrarily to elucidate \algName's behavior clearly.}

\begin{table*}[!t]
	\caption{\TOMM{Two probability distributions used in Section~\ref{sec:ex_simple}} }
	\label{tbl:ex1_probs}
	\begin{center}
		\begin{tabular}[P]{|c|c|c|c|c|c|c|}
			\hline
			\textbf{Distribution} & \multicolumn{6}{c|}{Probabilities of head direction in an descending order}\\
			\hline
			Homogeneous & 0.166 & 0.166 & 0.166 & 0.166 & 0.166 & 0.166 \\
			\hline
			Heterogeneous & 0.317 & 0.257 & 0.197 & 0.136 & 0.076 & 0.017 \\
			\hline
		\end{tabular}
	\end{center}
\end{table*}

Figure~\ref{fig:mic_rates2} shows the maximum, minimum, and average bitrates of downloaded tiles for each chunk of the video. For the homogeneous distribution, the selected bitrate for all tiles of a chunk is the same.
The results in Figure~\ref{fig:mic_rates} show that the average download bitrate grows with an increase in buffer level. We show the threshold values for the buffer level where the action for the tile with the highest probability changes.  In addition, we show the variations of buffer level over time for both homogeneous and heterogeneous head position probability distributions in Figure~\ref{fig:mic_rates}. When the buffer level is higher than $V \delta (v_M \cdot p_{k, d}  + \gamma \delta)$, \algName downloads nothing for that tile. Note that increasing the value of $\gamma$ increases the importance of continuous playback. Increasing the value of $\gamma$ by $\epsilon$ is similar to reducing the buffer level by $\epsilon \delta^2 V$, resulting in \algName  using correspondingly higher threshold values for the buffer levels for bitrate switches. Therefore, increasing the value of $\gamma$ shifts the bitrate curves in Figure~\ref{fig:mic_rates2} to the right and vice versa. Lastly, Figure~\ref{fig:undr_bitrate} and Figure~\ref{fig:undr_bitrate2} show the average bitrates of tiles downloaded across the time and the bitrate of the tiles that user actually sees (playing bitrate) in their \AF. One can see that \algName responds to the bandwidth change by increasing/decreasing selected bitrates.

\begin{figure*}
	\centering
    \subfigure{\includegraphics[width=0.36\textwidth]{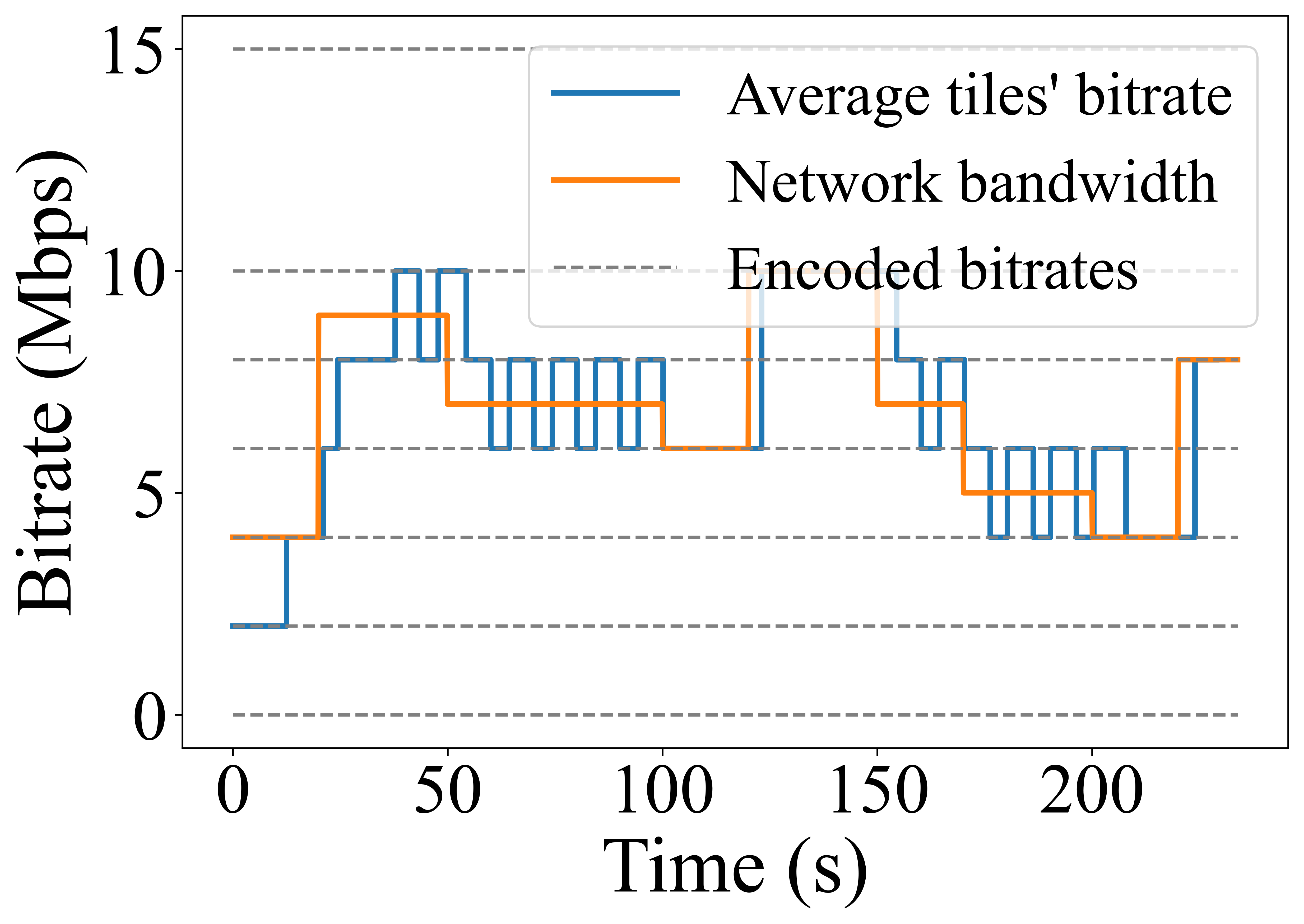}}\hspace{8mm}
	\subfigure{\includegraphics[width=0.36\textwidth]{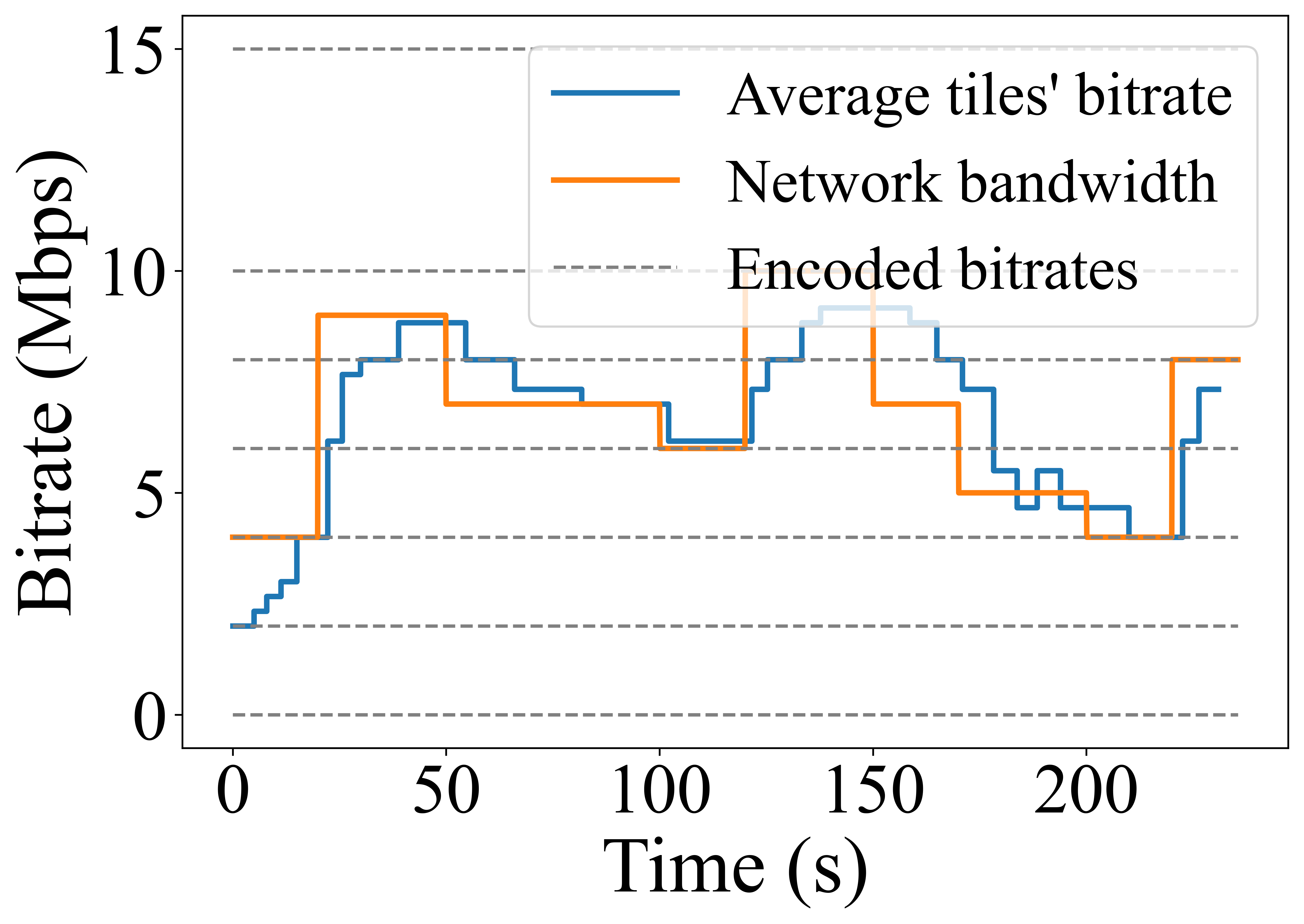}}
	\caption{Variation of average downloaded bitrates over time under bitrate selection of \algName for the homogeneous (left), and heterogeneous (right) head position probability distribution. }
	\label{fig:undr_bitrate}
\end{figure*}

\begin{figure*}
	\centering
	\subfigure{\includegraphics[width=0.36\textwidth]{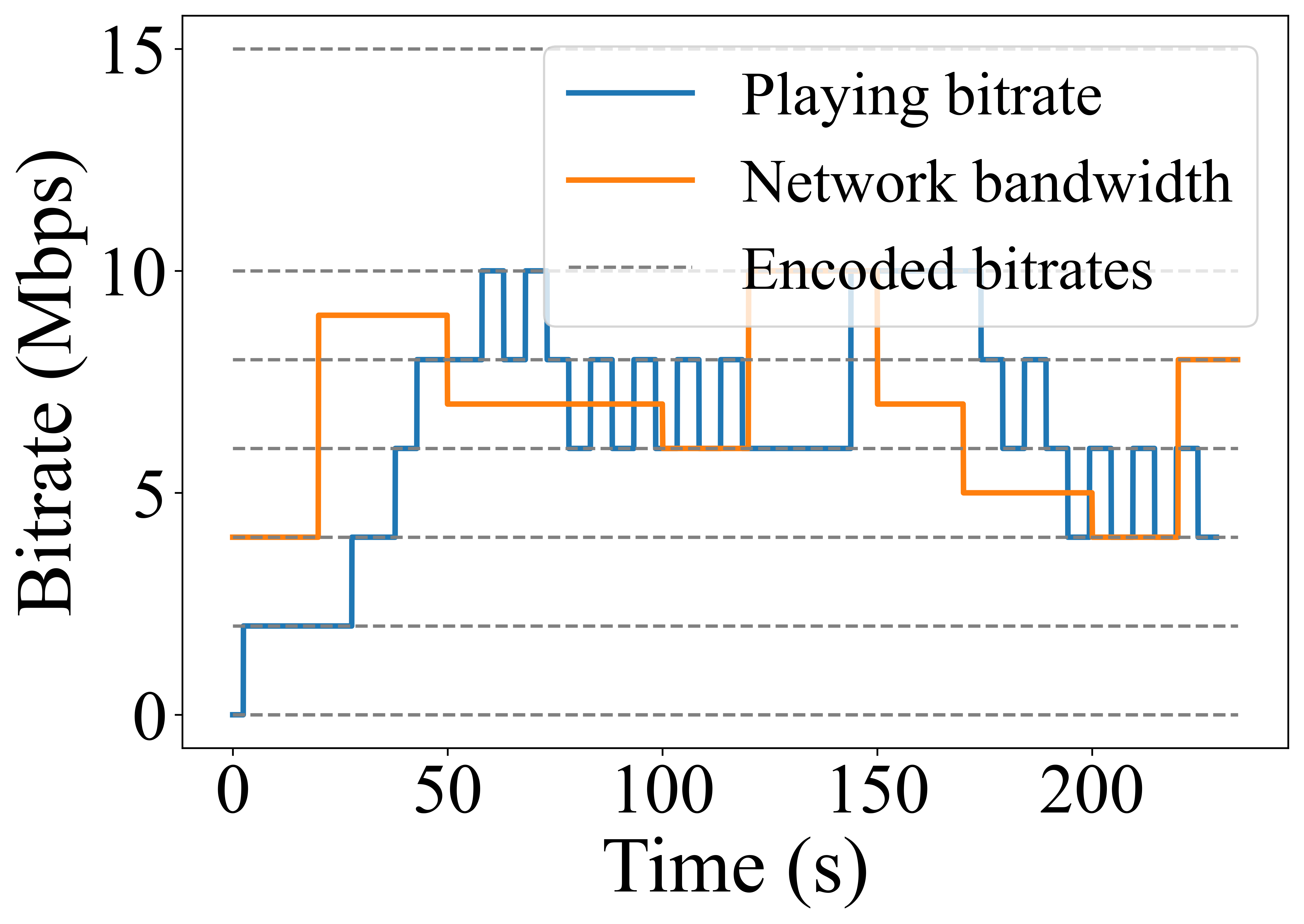}}\hspace{8mm}
	\subfigure{\includegraphics[width=0.36\textwidth]{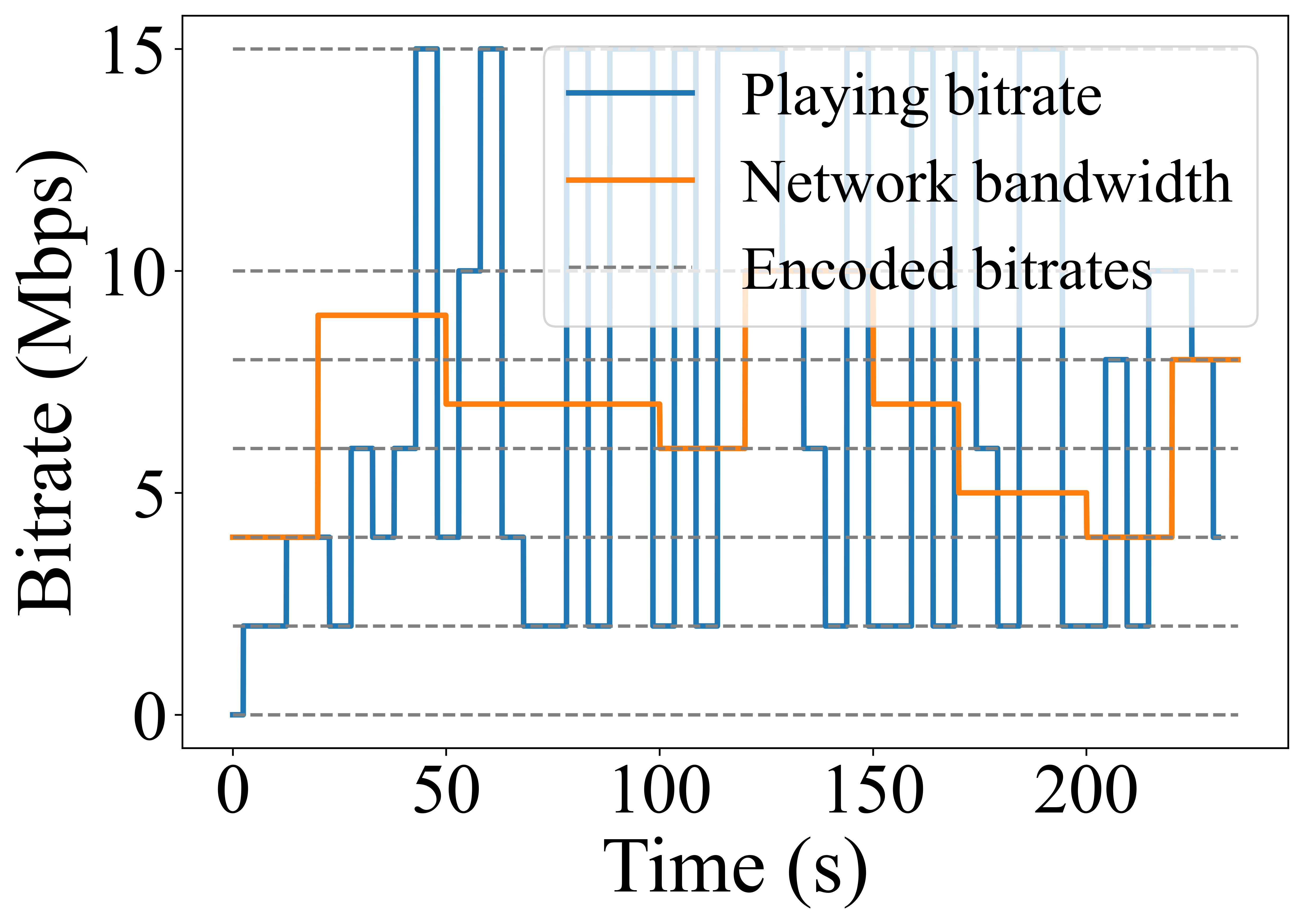}}
	\caption{Variation of playing bitrate over time under bitrate selection of \algName for the homogeneous (left), and heterogeneous (right) head position probability distribution. }
	\label{fig:undr_bitrate2}
\end{figure*}

\section{Comparison Algorithms}
\label{sec:exp}



We compare \algName with \VAts~\cite{ozcinar2017viewport}, \PDash~\cite{xie2017360probdash}, \Flare~\cite{qian2018flare}, \SalVR~\cite{wang2022salientvr}, \Pano~\cite{guan2019pano}, and \Mosaic~\cite{park2021mosaic}, the leading ABR algorithms for \tABR. Our analysis showcases the advancements our approach brings over the state-of-the-art. While some of these algorithms like \Pano, \Flare, and \PDash consider additional factors such as minimizing bitrate variance among tiles within a chunk, they rely on an MPC algorithm \cite{yin2015control} to select the aggregate bitrate for each chunk. This method's reliance on estimated bandwidth throughput poses challenges, potentially hindering their ability to achieve near-optimal QoE, a limitation shared by algorithms like \VAts and \Mosaic. We used the suggested hyper-parameters from each algorithm's respective literature.

\subsection{Experimental Setup}
\label{sec:exp_setup}
We conduct multiple experiments to demonstrate the algorithms' performance under different settings. We use a 500-second video, split into chunks of $2$ seconds and $8$ tiles. Also, \revised{each tile is encoded} in seven different bitrates - 440Kbps, 700Kbps, 1.35Mbps, 2.14Mbps, 4.1Mbps, 8.2Mbps, and 16.5Mbps. Similar to Section \ref{sec:ex_simple}, we use a logarithmic utility function. The list of available bitrates, size of segments, and utility values are listed in Table~\ref{tbl:realExp_bitrates}. \revised{Although \algName performs better using larger buffers, we limit the buffer capacity to $Q_{\max} = 128\delta$, which is equivalent to $32$ seconds of \threesixty playback time, that falls within the range of suggested buffer capacity for VOD streaming \cite{huang2014buffer, hao2021buffer} to ensure fairness.
We employ dynamic value selection for $\Delta$, as proposed in~\cite{spiteri2020bola}, and empirically determine $V = 24.0$. Also, we set $\gamma = 0.2$ using the parameter selection methodology for $\gamma$ and $V$ outlined in Section VI of~\cite{spiteri2020bola}. Consequently, in this scenario, the smoothness term of QoE is equivalent to the utility derived from downloading tiles at a bitrate of 2.1Mbps (equivalent to 480p resolution).} We use 4G bandwidth traces from~\cite{bokani2016prehensive} and 4G/LTE bandwidth trace dataset \cite{IDLAB} collected by IDLAB \cite{vanderHooft2016} to simulate the network condition. \revised{We select 14 different traces (network trace index 1 to 14 of the dataset) from 4G/LTE dataset} to evaluate the performance of \algName under different network conditions.
The video is stored on an Apache server. Both server and client use Microsoft Windows, 24GB of RAM, and an 8-core 3Ghz Intel Core-i7 CPU. We use Chrome DevTools API \cite{chrmDev} to transfer the video between server and client and emulate the network condition. We fetch the bandwidth capacity from the 4G/LTE dataset and inject it into the Chrome DevTools to limit the download capacity between the server and the client. In our experiments, \AF includes a single tile and unless otherwise mentioned, to capture the actual \AF of the head position probability values, we generate the navigation graph \cite{park2019navigation} for \threesixty video using public VR head traces~\cite{wu2017dataset}.

\begin{table*}[!t]
	\caption{\TOMM{Available bitrates and utility values of them for experiments of Sections \ref{sec:Exp2}, and \ref{sec:Exp3}} }
	\label{tbl:realExp_bitrates}
	\begin{center}
		\begin{tabular}[P]{|L{3cm}|c|c|c|c|c|c|c|}
			\hline
			\textbf{bitrate (Mbps)} & 0.44 & 0.7 & 1.35 & 2.14 & 4.1 & 8.2 & 16.5 \\
			\hline
			\textbf{Sizes (Mb)} & 0.88 &  1.4 &  2.7 & 4.28 & 8.2 & 16.4  & 33.0 \\
			\hline
			\textbf{Utility values} & 1.000 & 1.667 & 2.617 & 3.282 & 4.220 & 5.220 & 6.229 \\
			\hline

		\end{tabular}
	\end{center}
\end{table*}

\begin{figure*}
	\centering
    \hspace{0.025\textwidth}
    \subfigure{\includegraphics[width=0.95\textwidth]{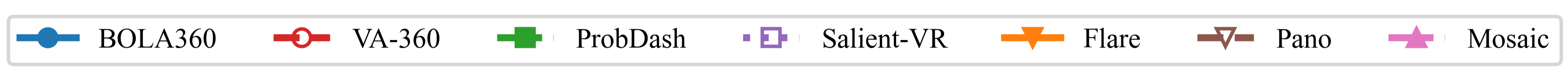}}\vspace{-2mm}
    \subfigure{\includegraphics[width=0.48\textwidth]{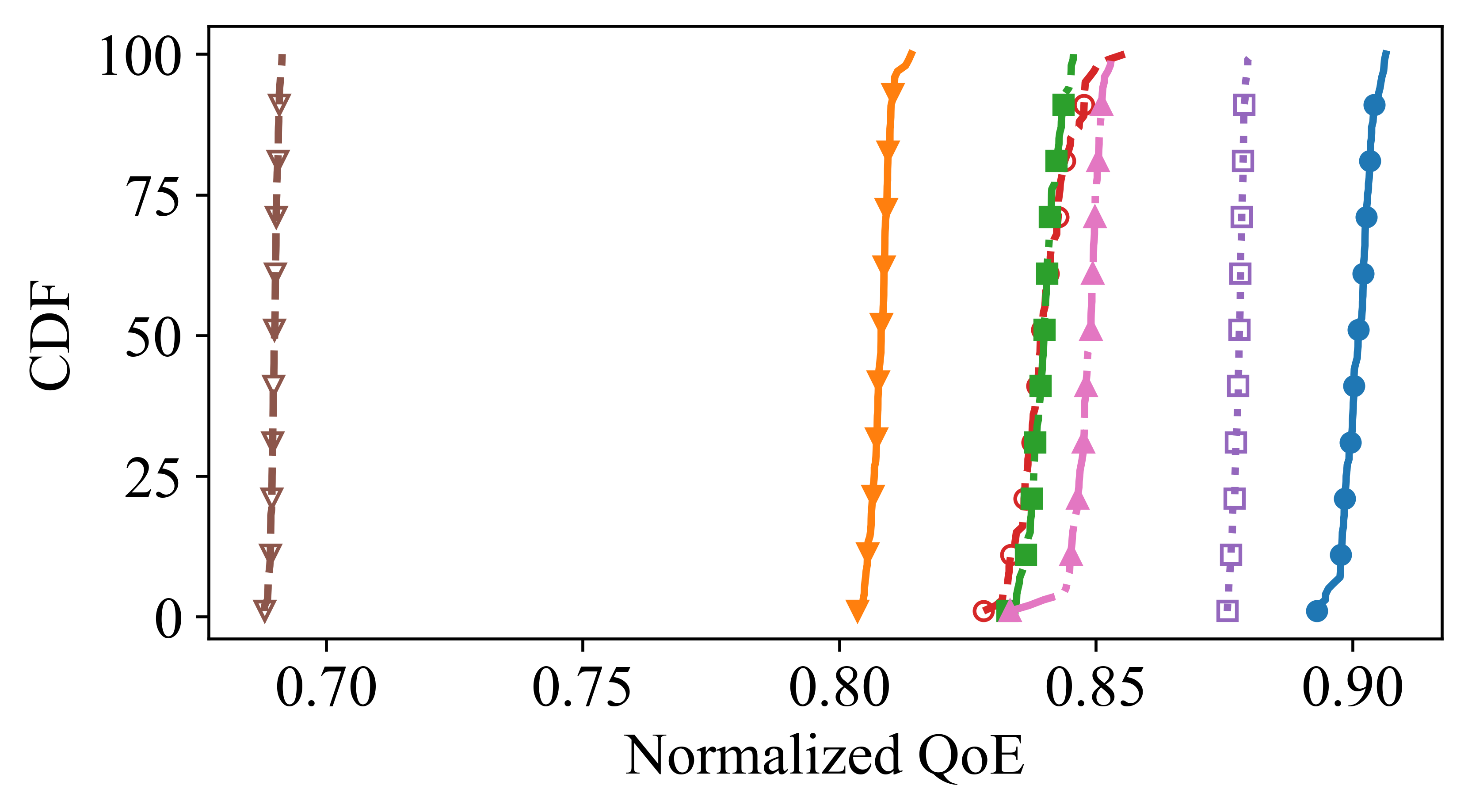}}\hspace{2mm}
    \subfigure{\includegraphics[width=0.48\textwidth]{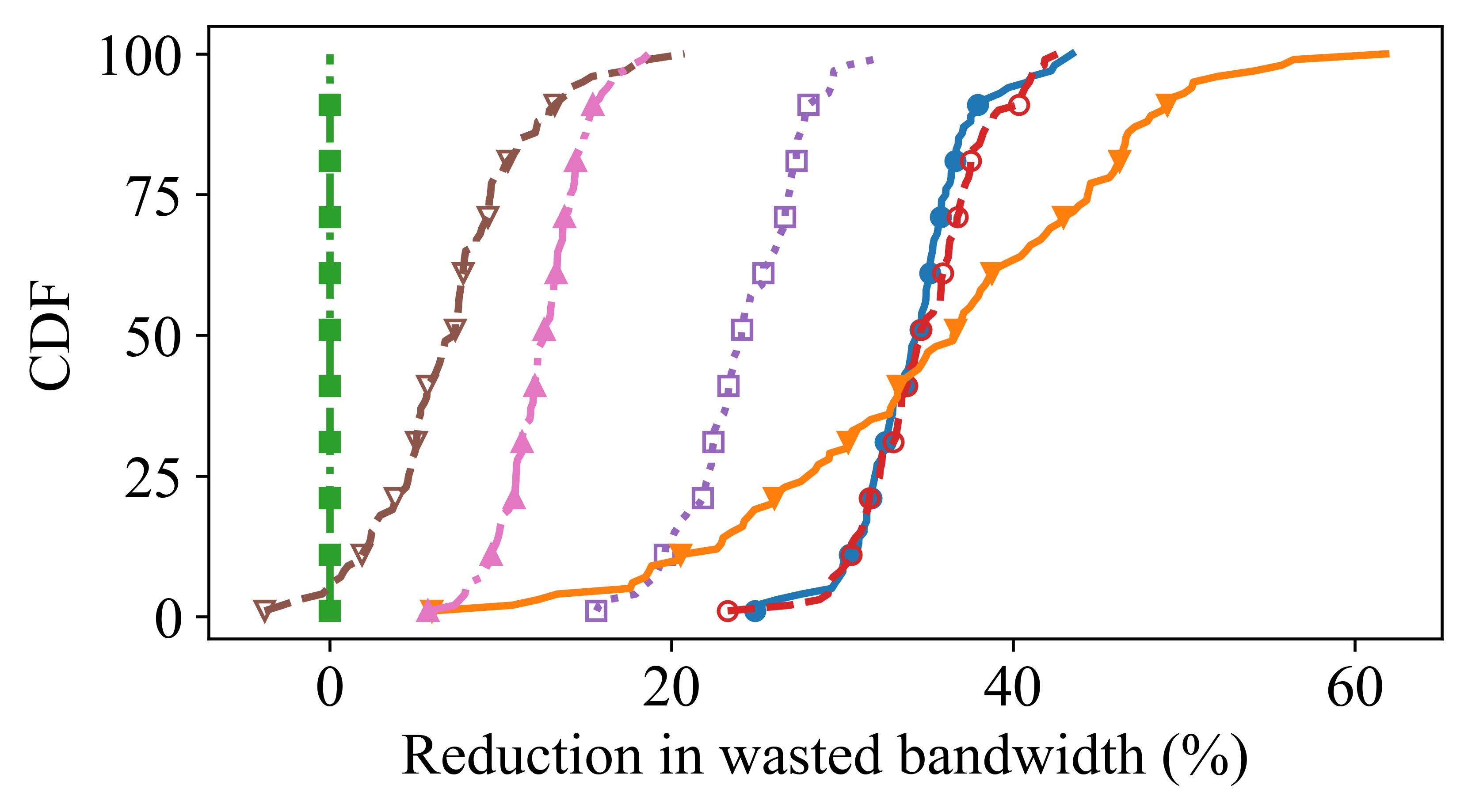}}

	\caption{\TOMM{The CDF of normalized QoE (left) and reduction in wasted bandwidth compared with the \PDash (right)} for \algName and comparison algorithms using real network and head movement traces. The QoE of \algName was higher than the QoE of other algorithms in all test trials. }
	\label{fig:realExp_QoE_wasted}
\end{figure*}

\subsection{Performance Evaluation using Real Network and Head Movement Traces}
\label{sec:Exp2}
First, we compare the performance of \algName with others using real network and head movement traces. We use a representative real 4G bandwidth trace from~\cite{bokani2016prehensive} for this comparison. We report playing bitrate, the rebuffering ratio (percentage of length of video considered as a rebuffering), and normalized QoE (QoE divided by the QoE of optimal offline algorithm) of \algName, and state-of-the-art algorithms. \TOMM{Additionally, we report the reduction in wasted bandwidth—the fraction of bandwidth used to download unseen tiles—for all algorithms, relative to \PDash, which exhibited the highest wasted bandwidth among the compared algorithms.}  Note that the average playing bitrate reported in Figure~\ref{fig:realExp_performances} is calculated over the tiles the user has seen inside \AF. We report the results of 100 different trials, where for each trial, we sample the user's head direction from the head position probability distribution and use the same network traces and algorithm parameters. The CDF plot of average playing bitrates, rebuffering ratio, and normalized QoE values of 100 different trials is reported in Figure~\ref{fig:realExp_QoE_wasted} and Figure~\ref{fig:realExp_performances}. 



\begin{figure*}
	\vspace{-5mm}
	\centering
    \hspace{0.025\textwidth}
    \subfigure{\includegraphics[width=0.95\textwidth]{figures/header.png}}\vspace{-2mm}
    \subfigure{\includegraphics[width=0.32\textwidth]{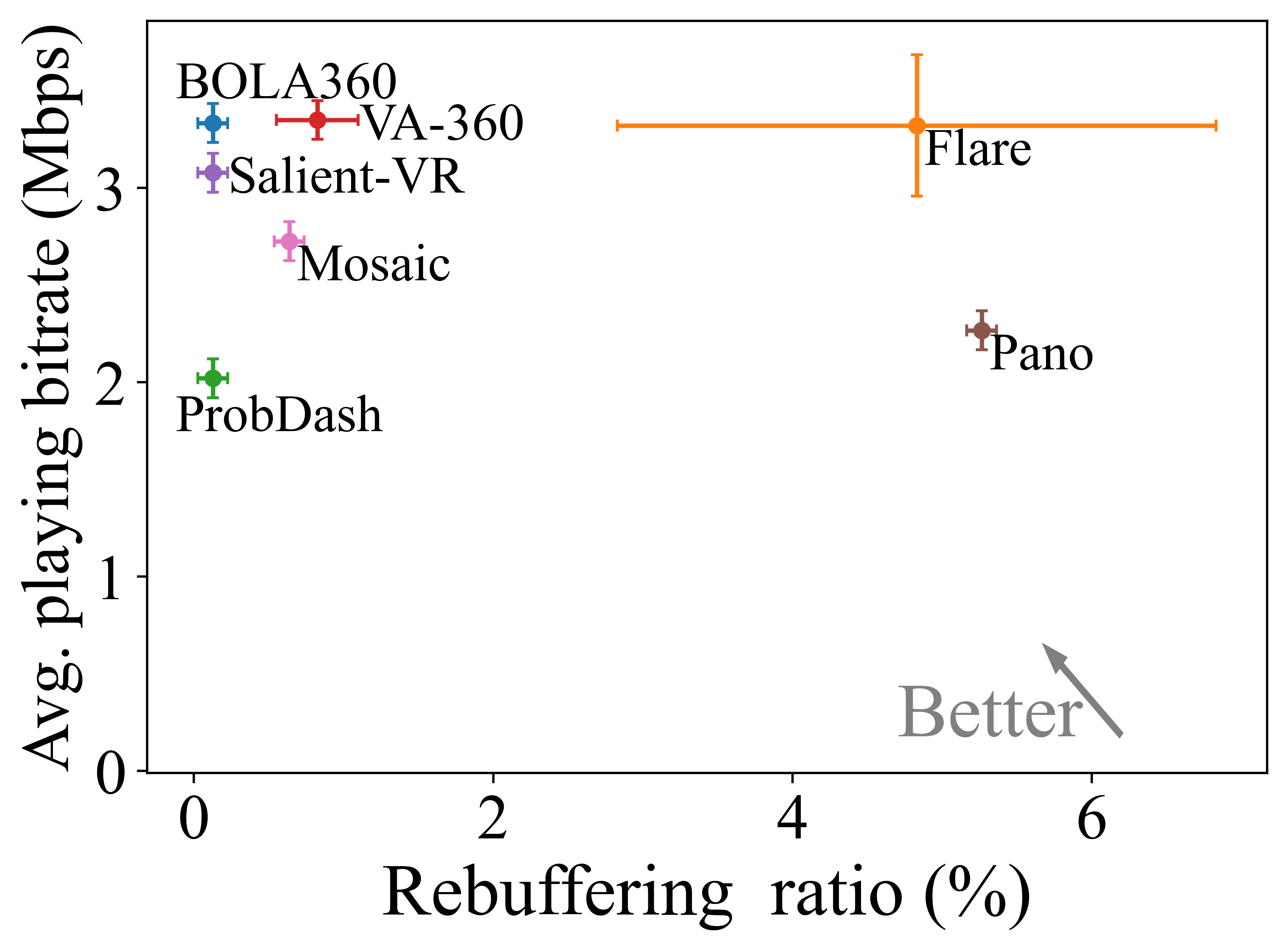}}
	\subfigure{\includegraphics[width=0.32\textwidth]{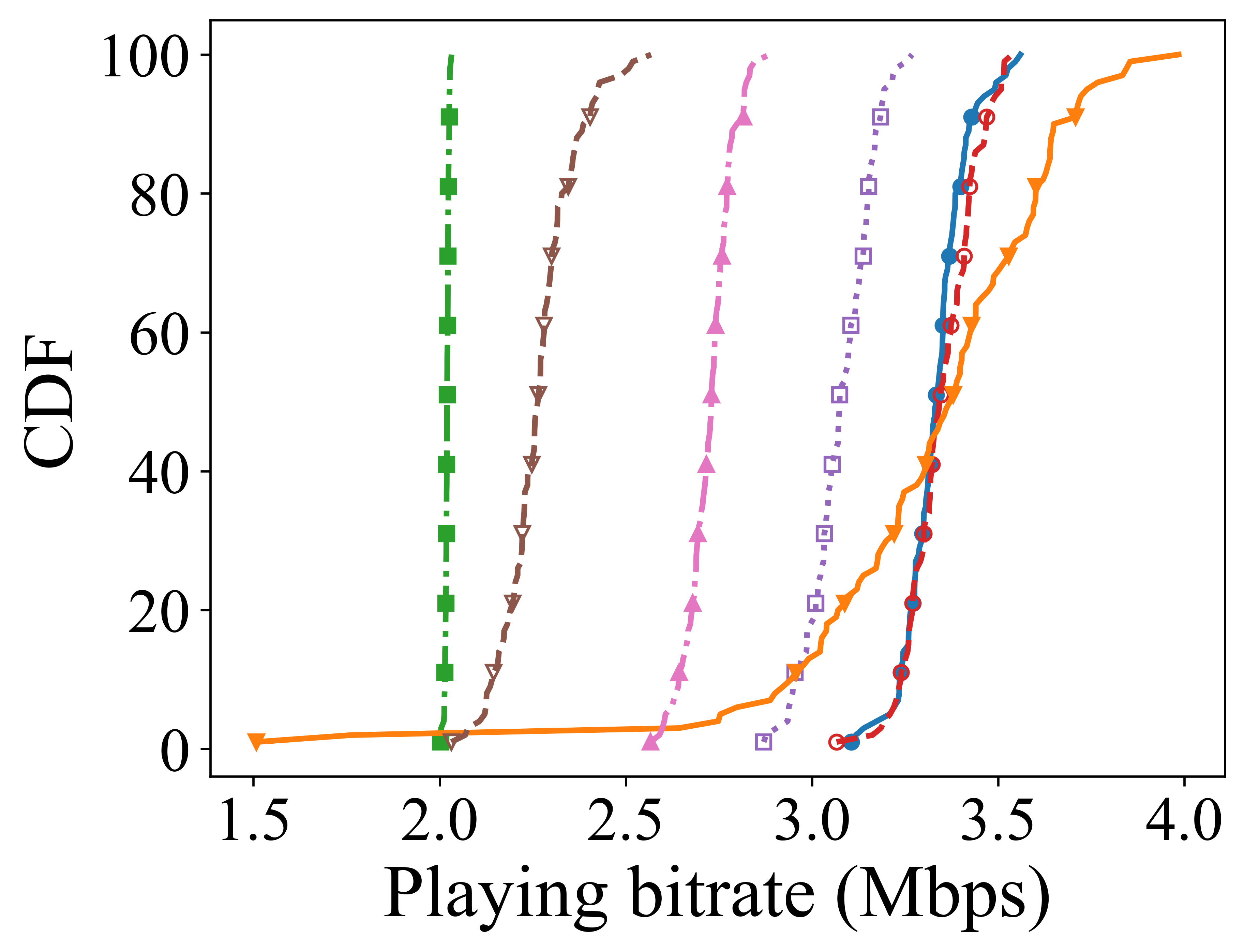}}
	\subfigure{\includegraphics[width=0.32\textwidth]{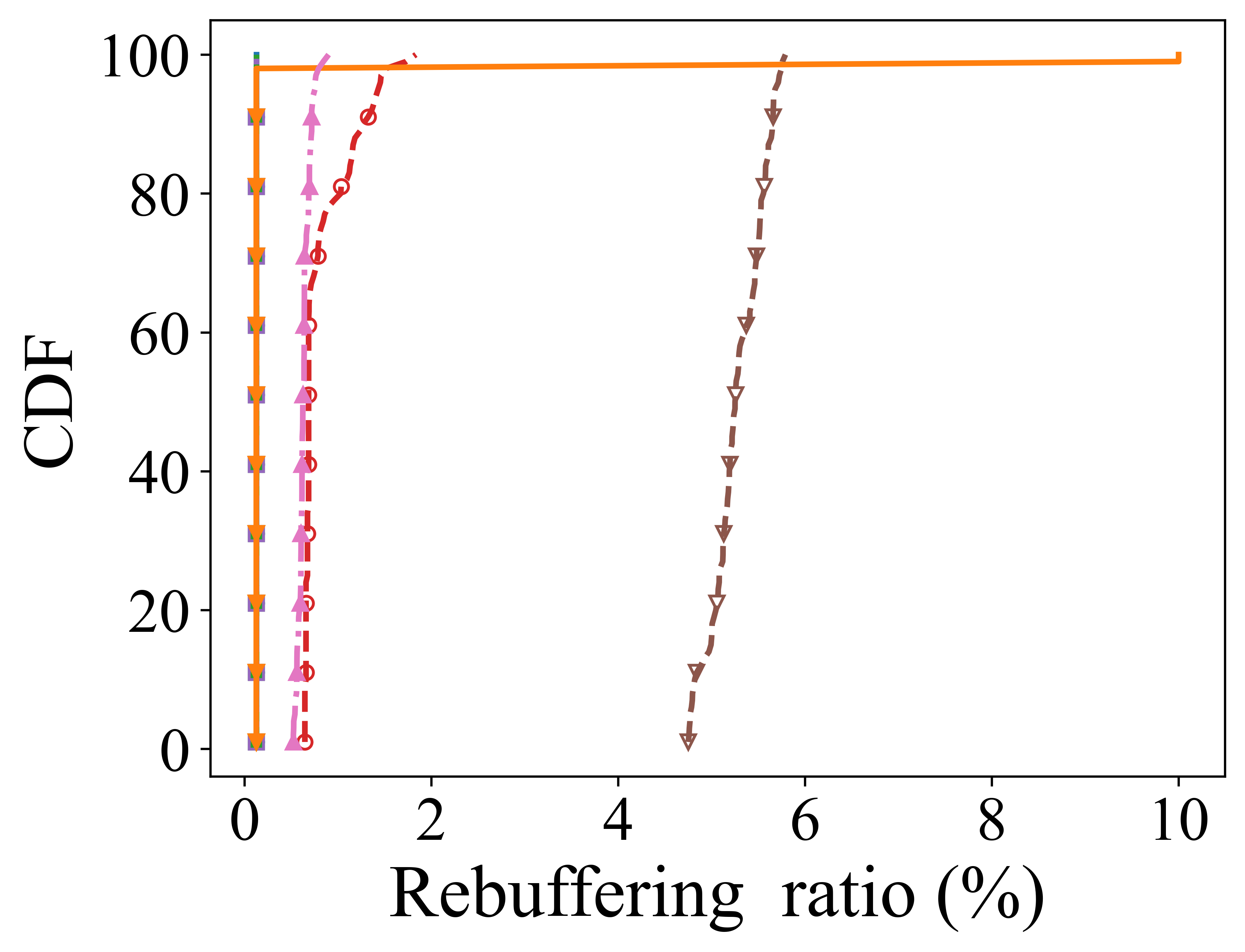}}
	\caption{Average playing bitrate vs. rebuffering ratio (left), the CDF of playing bitrate (midle), and rebuffering ratio (right) of \algName and comparison algorithms using real network and head movement traces. The average playing bitrate of \algName is $3.34Mbps$, while this value for \SalVR, and \Mosaic and \VAts are  $3.07Mbps$, $2.72Mbps$, and $3.35Mbps$. The average rebuffering for \algName, \SalVR, \Mosaic, and \VAts were $0.12\%$, $0.13\%$, $0.63\%$ and $0.83\%$. }
	\label{fig:realExp_performances}
    \vspace{-3mm}
\end{figure*}

The results in Figures ~\ref{fig:realExp_QoE_wasted} and~\ref{fig:realExp_performances} show that \algName outperforms other comparison algorithms in QoE, and its playing bitrate was slightly less than the playing bitrate of \VAts, which prepares the highest playing bitrates among comparison algorithms. \VAts selects relatively high bitrates for all tiles of a chunk while \algName efficiently distributes the available bitrates among different tiles such that \algName can achieve a lower rebuffering ratio. \TOMM{Additionally, the reduction in wasted bandwidth compared to \PDash is illustrated in Figure~\ref{fig:realExp_QoE_wasted}. On average, \algName achieves a $34.3\%$ reduction in wasted bandwidth compared to \PDash.}

\textbf{Key takeaway.} \algName outperforms comparison algorithms in terms of QoE as it is designed to maximize it. Besides, no algorithm outperforms \algName on both playing bitrate and rebuffering ratio at the same time.


\subsection{Impact of Network Bandwidth}
\label{sec:Exp3}

In this experiment, we investigate the impact of different network profiles on the performance of ABR algorithms. We use  network traces index 1 to 14 from the 4G/LTE dataset~\cite{IDLAB} to generate the bandwidth throughput. We use the same video and algorithm/problem parameters (details in Section \ref{sec:exp_setup}) for all algorithms to capture the impact of the network capacity on their performance. 

Figure \ref{fig:network_QoE_bitrate} shows the average normalized QoE, playback delay, rebuffering ratio, and average playing bitrate of \algName and five comparison algorithms over 100 trials for 14 network profiles. \algName stands as the best algorithm in all 14 experiments. In this experiment, \VAts selects relatively higher bitrates compared to other algorithms, while its high rebuffering, shown in Figure \ref{fig:network_rebuff_d2r}, lowers the QoE of this algorithm. In addition, the playback delay of \algName and comparison algorithms are shown in Figure \ref{fig:network_rebuff_d2r}. The playback delay of \VAts was the lowest in all experiments. That clearly shows the trade-off between having low rebuffering or low playback delay. The results show that \algName keeps the playback delay under $14.9$ seconds with an average rebuffering ratio of less than $0.4\%$. This playback delay is consisent with the result of Theorem~\ref{thm:buffer_size} and the fact that \algName tries to keep the buffer level high. \TOMM{Additionally, the reduction in wasted bandwidth, as compared to \PDash, for \algName and other comparison algorithms over 100 trials across 14 network profiles is illustrated in Figure~\ref{fig:network_wasted}. \Flare achieves the highest reduction in wasted bandwidth, decreasing it by an average of $58.5\%$. For \algName and \Pano, the reductions were $44.9\%$ and $32.5\%$, respectively.}

\textbf{Key takeaway.}  Networks with high fluctuations (e.g., profile indexes 2 and 7) cause a higher rebuffering ratio; nevertheless, \algName keeps QoE and playing bitrate relatively high in all profiles and outperforms all alternatives.

\begin{figure*}
	\centering
    \subfigure{\includegraphics[width=\fixWideWidth \linewidth, height=\fixWideHeight \linewidth]{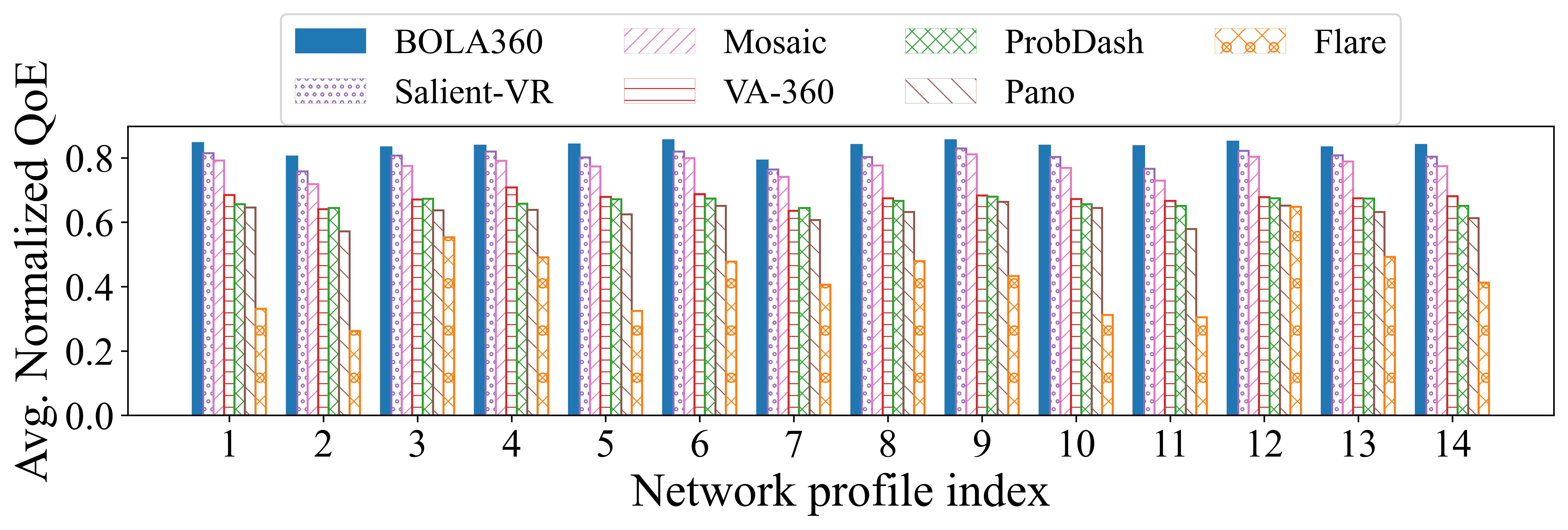}}
	\subfigure{\includegraphics[width=\fixWideWidth \linewidth, height=\fixWideHeight \linewidth]{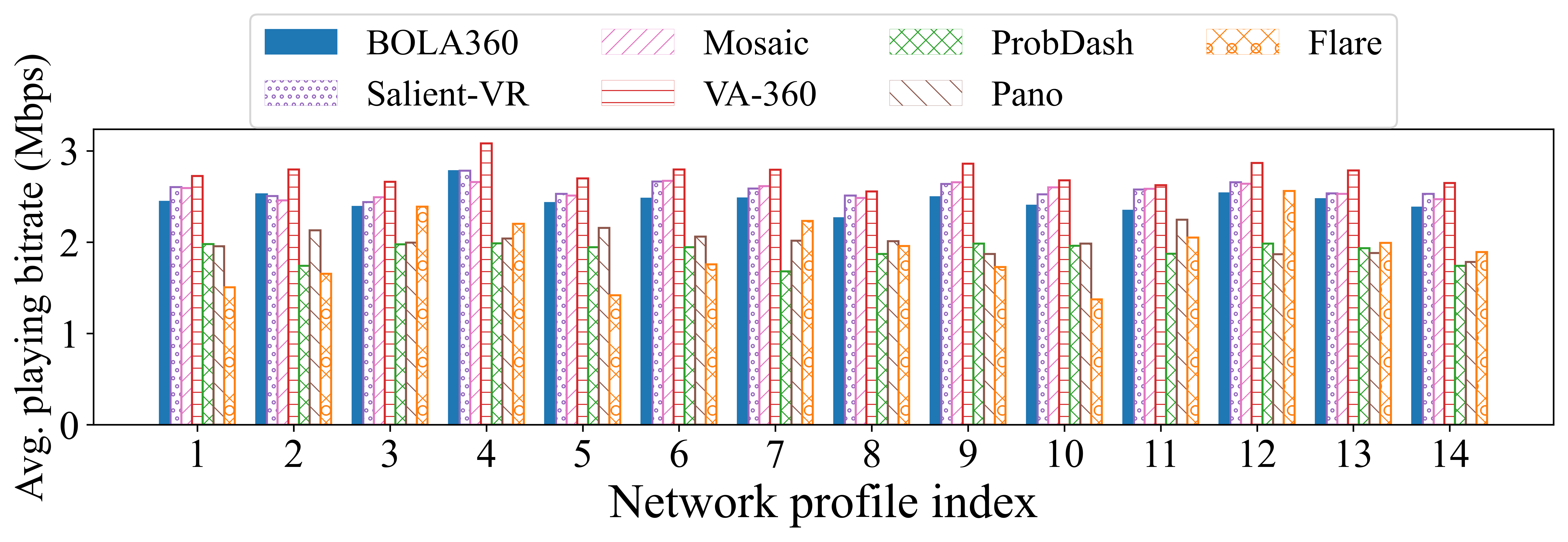}}
	
	\caption{The average normalized QoE (left) and average playing bitrate (right) over the bitrate selection of \algName and other comparison algorithms for 14 different network profiles and 100 trials. In terms of QoE, \algName outperforms others in all profiles. \camRev{On average, \algName provides about $6\%$ improvement to the QoE of \SalVR, and $110\%$ to QoE of \Flare.}}
	\label{fig:network_QoE_bitrate}
\end{figure*}



\begin{figure*}
	\centering
    \subfigure{\includegraphics[width=\fixWideWidth \linewidth, height=\fixWideHeight \linewidth]{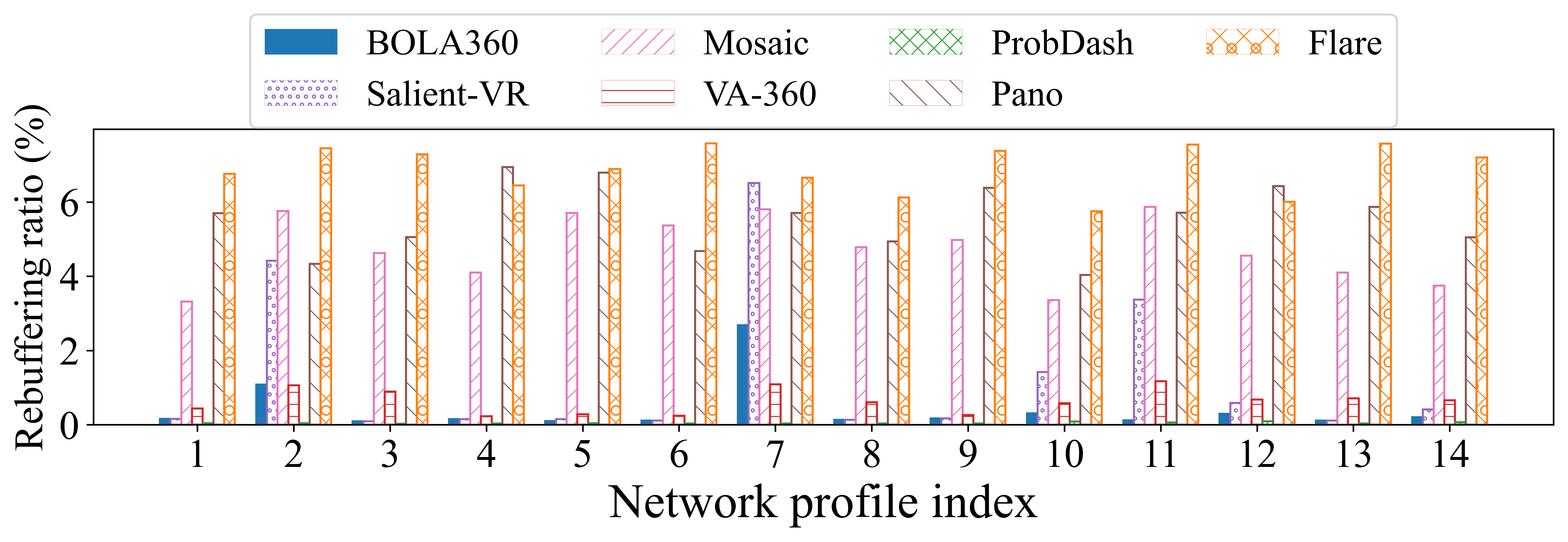}}
	\subfigure{\includegraphics[width=\fixWideWidth \linewidth, height=\fixWideHeight \linewidth]{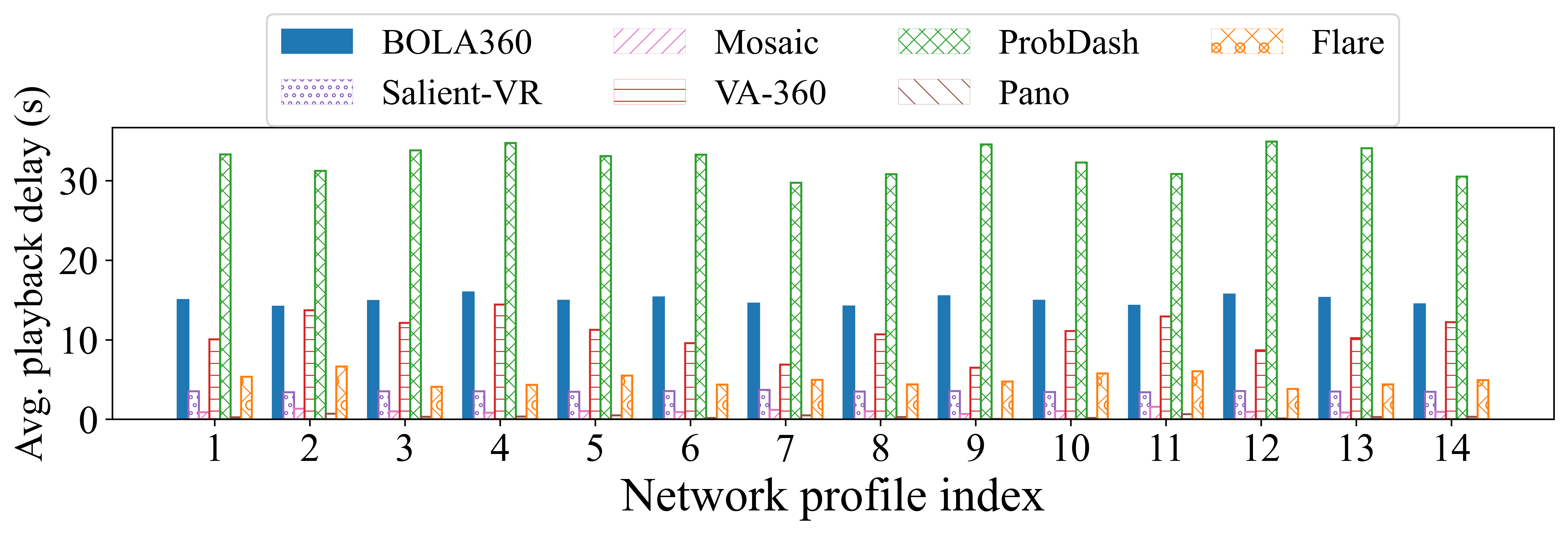}}
	\caption{The average rebuffering ratio (left) and average playback delay (right) over the bitrate selection of \algName and comparison algorithms for 14 different network profiles and 100 trials. \Pano and \Flare usually show higher rebuffering than the other algorithms, while their playback delay is shorter. The average playback delay for \algName is $14.9$ seconds. }
	\label{fig:network_rebuff_d2r}
\end{figure*}



\subsection{Impact of Head Position Probabilities}
\label{sec:Exp4}
The head position probability values directly impact the QoE characterized in Equations~\eqref{eq:UK} and~\eqref{eq:RK}; hence, the performance of algorithms varies depending on these probabilities. To observe the impact of head position probabilities on the performance of ABR algorithms, we define 12 probability distributions and evaluate the performance of \algName and other algorithms against them while the rest of the setting is similar to the experiment in Section \ref{sec:Exp2}. Specifically, for each chunk $k$, we replace the set of probabilities with the probabilities calculated from Equation~\eqref{eq:prob_profile_vals}. 
Each head position probability distribution could be interpreted as a different \threesixty video file. 


We generate the head position probability distributions based on three parameters $D_{\text{pos}}(k)$, $r(k)$, and $\alpha_p(k)$. Parameter $D_{\text{pos}}(k)$ shows the number of tiles that there is a chance to be inside \AF for chunk $k$; $r(k)$ represents the ratio between the minimum and maximum probabilities among probabilities of tiles for chunk $k$. Last, parameter $\alpha_p(k)$ determines the heterogeneity of the head position probability values for chunk $k$. We define the probability of $i^{th}$ highest probable tile as a function of $\alpha_p(k)$ as follows.
\begin{equation}
\label{eq:prob_profile_vals}
    p_{i}(k) =  \frac{1 - \alpha_p(k)} {D_{\text{pos}}(k)} + \alpha_p(k) p_{i}^{(L)}\bigg(k, r(k)\bigg),
\end{equation}
 \revised{where $p_{i}^{(L)}(k, r(k))$ shows the probability of $i^{th}$ highest probable tile assuming a fixed step between probabilities in ascending order.} With the above definition in Equation~\eqref{eq:prob_profile_vals}, $\alpha_p(k)$ determines the range of probabilities where $\alpha_p(k) = 0$ signifies uniform tile probabilities, while $\alpha_p(k) = 1$ indicates a wider probability range, reflecting diverse head position probability values. A justification for this model as a representative of real-world head direction prediction is that $(D - D_{\text{pos}})$ shows the number of tiles the \AF prediction model is confident that they will be out of \AF. On the other hand, $\alpha_p(k)$ shows how concentrated the \AF prediction model is. \revised{We use $r(k) = 0.05$ for all distributions. Although it's impractical to cover every possible distribution, our selection involves a low value for $r(k)$, and wide range of values for $D_{\text{pos}}$ and $\alpha_p(k)$ to achieve broader representation.} Details of the 12 probability distributions used in this section are outlined in Table \ref{tbl:prob_profiles}.
\begin{figure}[!t]
    \centering
    \begin{minipage}{\fixWideWidth \linewidth} 
        \centering
        \includegraphics[width=\textwidth]{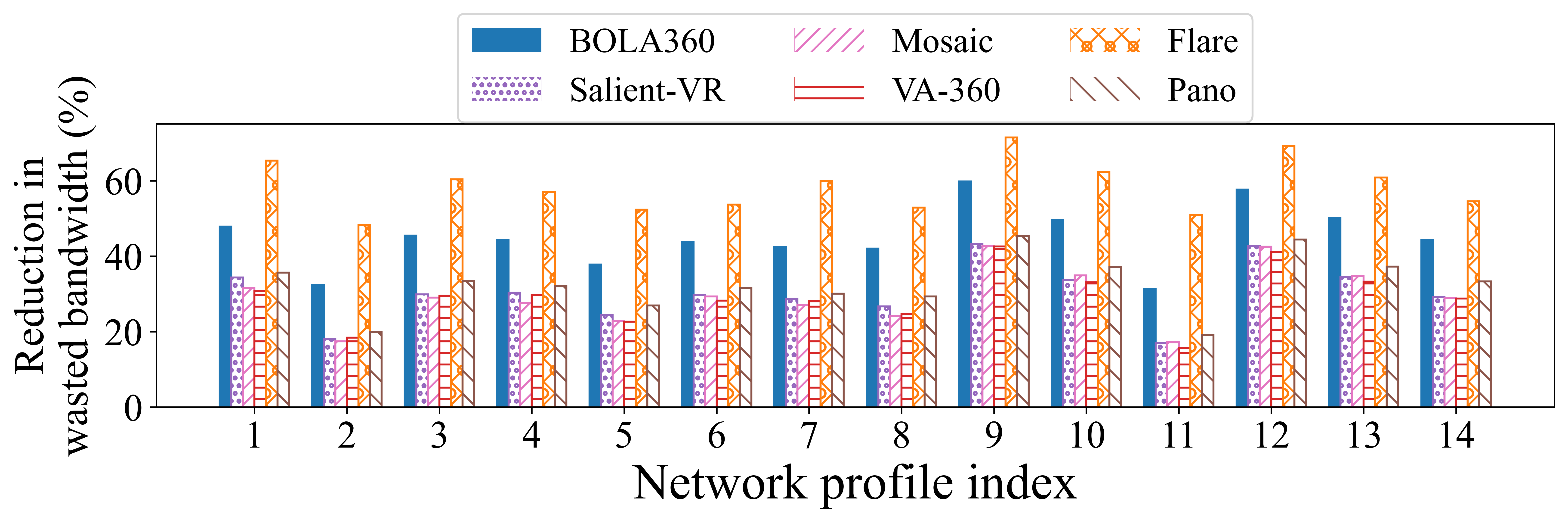}  
        \caption{\TOMM{Reduction in wasted bandwidth compared to \PDash over the bitrate selection of \algName and comparison algorithms using 14 different network distributions over 100 trials. The reduction of the wasted bandwidth under bitrate control of \algName was between $31.3\%$ and $59.8\%$ over different network profiles.}}
        \label{fig:network_wasted}
    \end{minipage}
    \hfill
    \begin{minipage}{\fixWideWidth \linewidth}  
        \centering
        \includegraphics[width=\textwidth]{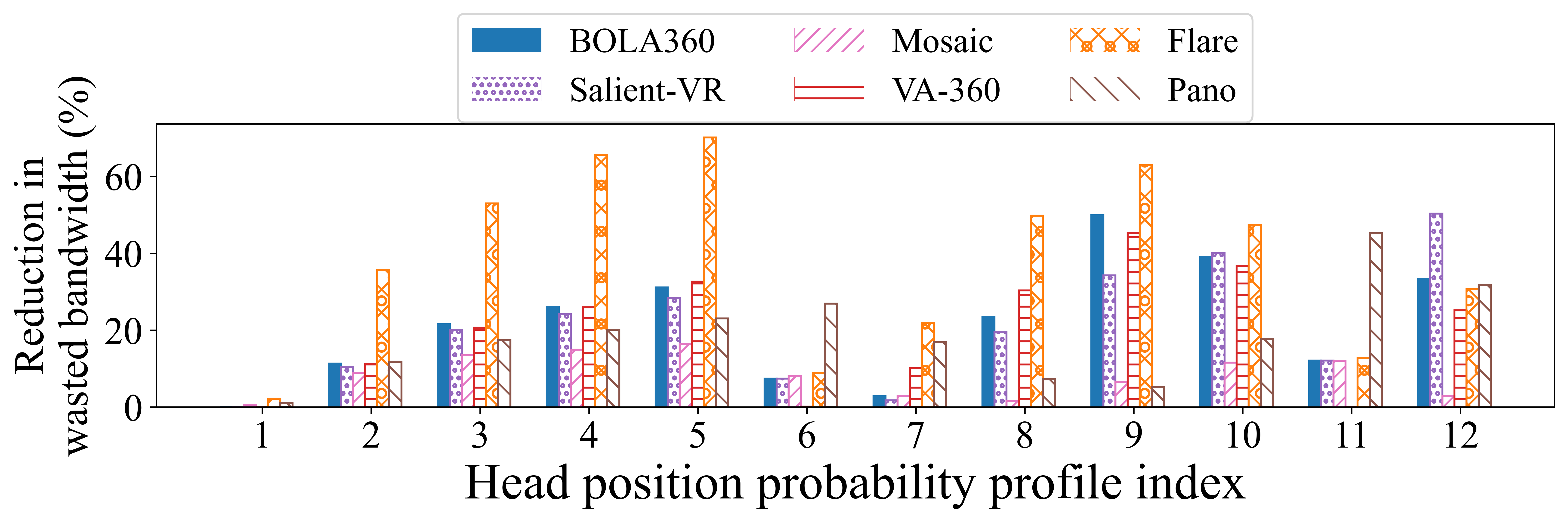}  
        \caption{\TOMM{Reduction in wasted bandwidth compared to \PDash over the bitrate selection of \algName and comparison algorithms using 12 different head position probability distributions over 100 trials. The reduction of the wasted bandwidth under bitrate control of \algName was between $1.2\%$ and $49.9\%$.}}
        \label{fig:prob_wasted}
    \end{minipage}
\end{figure}
 
\begin{table*}[!t]
	\caption{The details of the probability distributions used in the experiment of Section~\ref{sec:Exp4}}
	\label{tbl:prob_profiles}
	\begin{center}
		\begin{tabular}[P]{|c|c|c|c|c|c|c|c|c|c|c|c|c|}
			\hline
			  \textbf{Probability profile index} &  \textbf{1}&  \textbf{2} &   \textbf{3} &   \textbf{4} &   \textbf{5} &   \textbf{6} &   \textbf{7} &   \textbf{8} &   \textbf{9} &   \textbf{10} &   \textbf{11} &   \textbf{12}\\
            \hline
            $D_{\text{pos}}(k)$ & 8 & 8 & 8 & 8 & 8 & 4 & 4 & 4 & 4 & 4 & 2 & 2\\
            \hline
            $\alpha_p(k)$ & 0 & 0.25 & 0.5 & 0.75 & 1 & 0 & 0.25 & 0.5 & 0.75 & 1 & 0 & 0.5\\
            \hline

		\end{tabular}
	\end{center}
\end{table*}


We report the average normalized QoE, playback delay, rebuffering ratio, and average playing bitrate of 100 trials of \algName and comparison algorithms using each head position probability distribution profile in Figures~\ref{fig:prob_QoE_bitrate} (average normalized QoE and playing bitrate), and \ref{fig:prob_rebuff_d2r} (average rebuffering ratio and playback delay).  Figure~\ref{fig:prob_QoE_bitrate}  shows that \algName achieves slightly higher QoE when the prediction of \AF is concentrated on fewer number of tiles. A notable observation demonstrates that \algName kept the QoE at a high value for every probability profile, while the achieved playing bitrate is promising, and kept rebuffering ratio close to the lowest among all algorithms. \TOMM{Finally, Figure~\ref{fig:prob_wasted} presents the reduction in wasted bandwidth relative to \PDash for \algName and other comparison algorithms, based on 100 trials across 12 different head position probability distributions. The reduction in wasted bandwidth achieved by \algName ranged from $1.2\%$ to $49.9\%$, depending on the accuracy of \AF prediction across these head position probability profiles.}


\textbf{Key takeway.} The playing bitrate of \algName and most comparison algorithms improves when the head position prediction is concentrated on fewer tiles. Meanwhile, \algName improves the playing bitrate more than other comparison algorithms as the head position values concentrate on fewer tiles.

\begin{figure*}
	\centering
    \subfigure{\includegraphics[width=\fixWideWidth \linewidth, height=\fixWideHeight \linewidth]{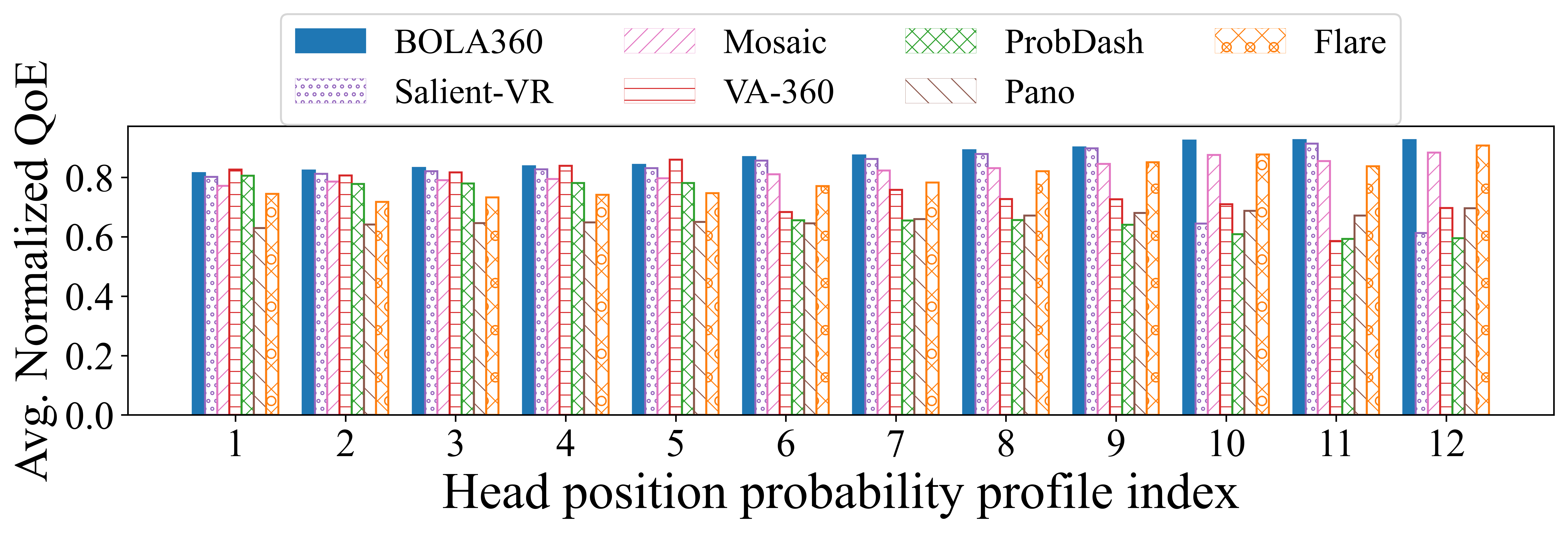}}
	\subfigure{\includegraphics[width=\fixWideWidth \linewidth, height=\fixWideHeight \linewidth]{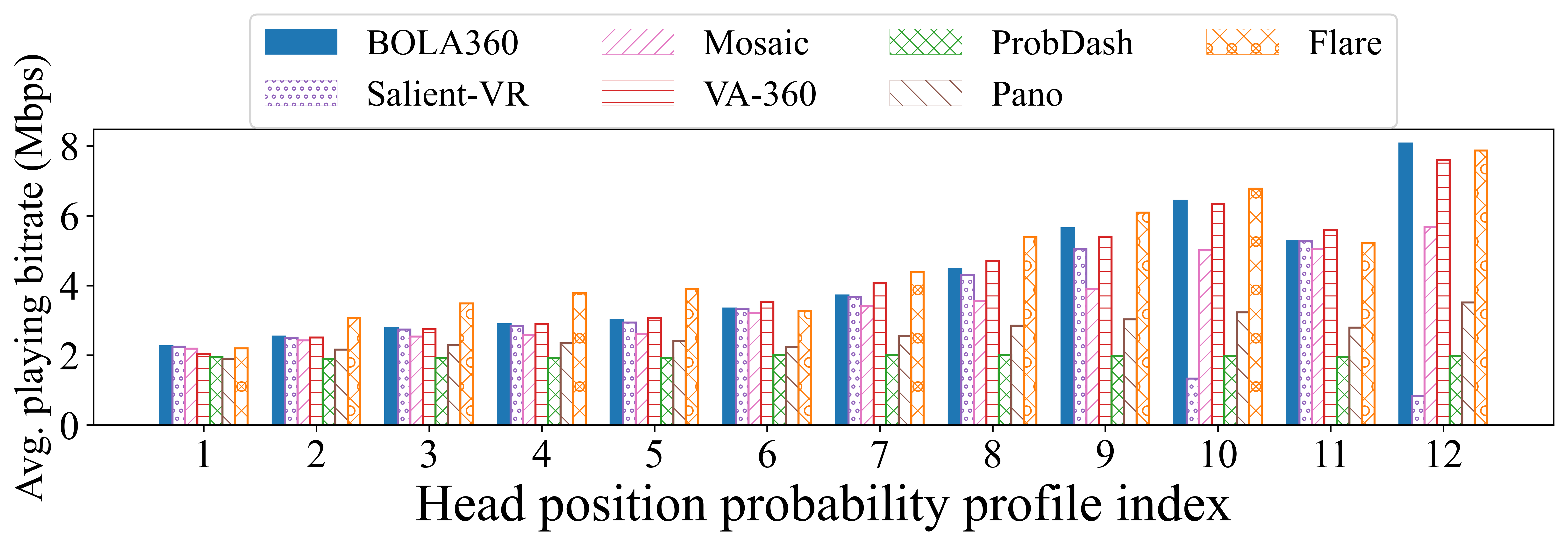}}
	\caption{The average normalized QoE (left) and average playing bitrate (right) over the bitrate selection of \algName and comparison algorithms using 12 different head position probability distributions over 100 trials. On average, \algName provides $9\%$ improvement to the QoE of \SalVR, and $31\%$ to QoE of \Pano. }
	\label{fig:prob_QoE_bitrate}
\end{figure*}



\begin{figure*}
	\centering
    \subfigure{\includegraphics[width=\fixWideWidth \linewidth, height=\fixWideHeight \linewidth]{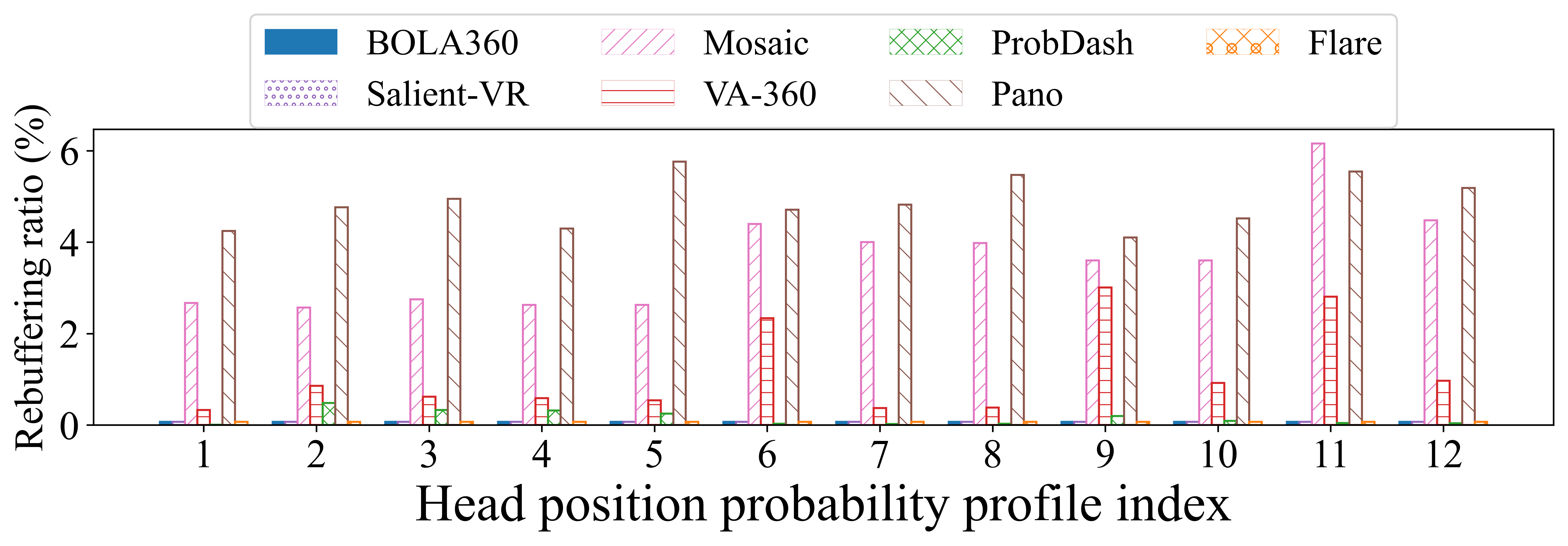}}
	\subfigure{\includegraphics[width=\fixWideWidth \linewidth, height=\fixWideHeight \linewidth]{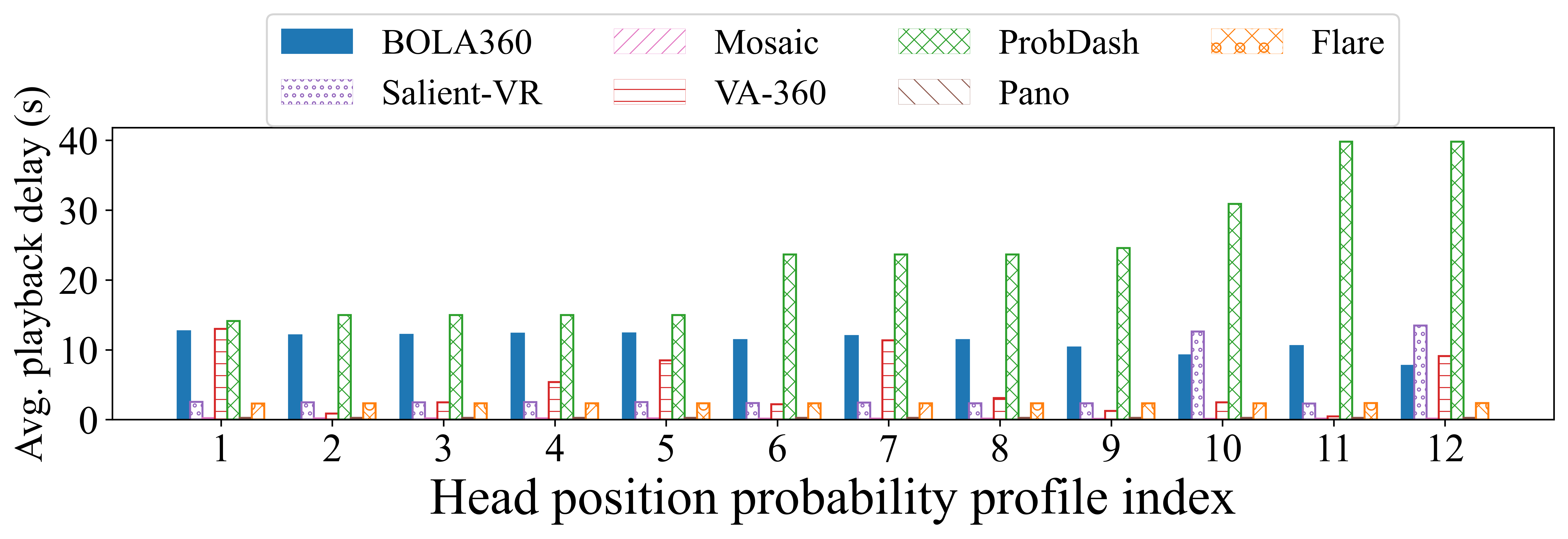}}
	\caption{The average rebuffering ratio (left) and average playback delay (right) over the bitrate selection of \algName and comparison algorithms using 12 different head position probability distributions over 100 trials. \VAts usually results in a high rebuffering ratio and short playback delay. Meanwhile, the rebuffering ratio of \algName is slightly more than the lowest in all experiments. The average playback delay for \algName is $11.1$ seconds. }
	\label{fig:prob_rebuff_d2r}
\end{figure*}

\subsection{Discussion on the Performance of Predictions-based Baseline Algorithms}
\label{sec:discussion}
\TOMM{This section provides details on why the baseline or state-of-the-art algorithm used in Section \ref{sec:exp_setup} may fail to perform well in particular scenarios, and they cannot guarantee their performance in the worst-case scenario. All of \VAts, \PDash, \Flare, \Pano, and \SalVR algorithms take action based on the prediction of bandwidth that is given to them. The accuracy of this prediction significantly impacts the performance of these algorithms such that a prediction with an error may result in a significant difference between the performance of the ABR algorithm and the performance of the optimal offline solution. In addition, these algorithms behave similarly to the \tABR with different values of $\gamma$. For example, for tiny values of $\gamma$, the bitrate level of the segments are much more important to the user than the smoothness of streaming. However, these  algorithms take similar actions as they take in the case of a large value of $\gamma$.}

\section{\algName Enhancements}
\label{sec:heuristics}
\algName is meticulously designed to excel under all conceivable network conditions, including the most challenging worst-case-like scenarios. The aim to achieve a satisfactory performance across all input, however, makes \algName often operate conservatively, refraining from switching to higher bitrates in many real-world situations where worst-case conditions fail to materialize. In this section, we propose \hrREP and \hrPL, two heuristic algorithms to improve the practical performance of \algName could be improved from two perspectives. First, we introduce \hrPL to address the common drawback of buffer-based ABR algorithms in fetching low-quality bitrates during start or seek time or high oscillations time intervals. Secondly, we propose \hrREP to add the tile upgrade into the \algName. The basic \algName algorithm is not designed to replace previously downloaded tiles with higher bitrates, further restricting its adaptability.

\hrPL is a generalized version of \texttt{BOLA-PL} introduced in~\cite{spiteri2019theory}. It aims to reduce the reaction time of the \algName during start and seek times by virtually increasing the buffer level at the start or seek time. The reaction time is the duration from when the first tile is fetched (during start time) or the first seek tile is fetched (during seek time) until bitrate of selected tiles stabilizes.
\hrPL estimates the bandwidth and multiplies it by $50\%$ to establish a safe expected bandwidth. To prevent rebuffering, \hrPL limits the bitrate of each tile based on the estimated bandwidth throughput. More specifically, it restricts the size of the entire chunk to $S_{lim} = Q(t) w_p(t) / 2D$, where $w_p(t)$ denotes the predicted bandwidth capacity at time $t$. \hrPL virtually inserts a proportional number of tiles into the buffer such that the size of the new downloading chunk does not exceed $S_{lim}$.


One of the primary limitations of the basic version of the \algName algorithm is its inability to modify previously downloaded tiles. Specifically, \algName must make decisions about tiles of future chunks, and it cannot replace higher bitrate tiles with previously downloaded, lower quality ones. If the bandwidth capacity experiences a short-term decrease, \algName adjusts the download bitrates to match the new bandwidth capacity by switching to lower bitrates. When the bandwidth increases again, \algName may have already downloaded several tiles with low bitrates, and it cannot replace them with higher quality ones, even if the buffer level and bandwidth capacity are high. Consequently, \algName cannot utilize the entire bandwidth opportunity to optimize QoE. This challenge is addressed by the heuristic called \texttt{BOLA360-REP}.


\TOMM{\hrREP is a variant of \algName that allows for the modification of previously downloaded tiles. Specifically, \hrREP determines whether to download tiles for the next chunk or improve the quality of previously downloaded tiles, depending on the available video length in the buffer. \hrREP uses a danger threshold set at $2 \delta$ and downloads tiles for a new chunk if the available video length in the buffer, $Q_{avl}(t)$, is below this threshold. If $Q_{avl}(t)$ exceeds the danger threshold, \hrREP replaces previously downloaded tiles with higher bitrates.}

\TOMM{When deciding to download tiles for a new chunk, \hrREP selects bitrates according to \algName's decisions. If the decision is to replace previously downloaded tiles, \hrREP identifies tiles where there is at least a two-level difference between the current bitrate and what \algName would select at the current buffer level. \hrREP then downloads and replaces those low-quality tiles. If no low-quality tiles are detected, \hrREP proceeds to download tiles for the next chunk as usual. Overall, \hrREP addresses the limitations of \algName and enhances the quality of experience for users.}

\subsection{Experimental Setup}

We use the parameters from Section~\ref{sec:Exp4} and the head position probability profile~2 defined therein to evaluate the performance of the heuristic extensions, \hrPL and \hrREP. These algorithms are tested under two scenarios: 1) accurate head position probability predictions, and 2) noisy predictions for future chunks. In the first scenario, the head position probabilities provided to the ABR algorithms match the actual head position distribution of the user. This implies that the algorithm has perfect knowledge of the user's head position distribution, even for chunks that will be played far in the future.

In contrast, the second scenario assumes a $10\%$ error in the prediction of head position probabilities for every $\delta$ seconds difference between the current chunk and the chunk for which the ABR algorithm is trying to predict head position probabilities. If the prediction error exceeds $100\%$, the head position prediction is considered unreliable, and the ABR algorithms are instead provided with uniform head position probabilities, where $p_{k,d} = 1/D$. In this case, if the ABR algorithm maintains a high buffer utilization, there would be a long period of time between download time and playback time of each chunk. This long period leads to noisy prediction of head position probabilities at the time of download, resulting in reduced QoE for the user.

\subsection{Experimental Results}
\begin{figure*}[t]
	\centering
	\subfigure{\includegraphics[width=0.32\textwidth]{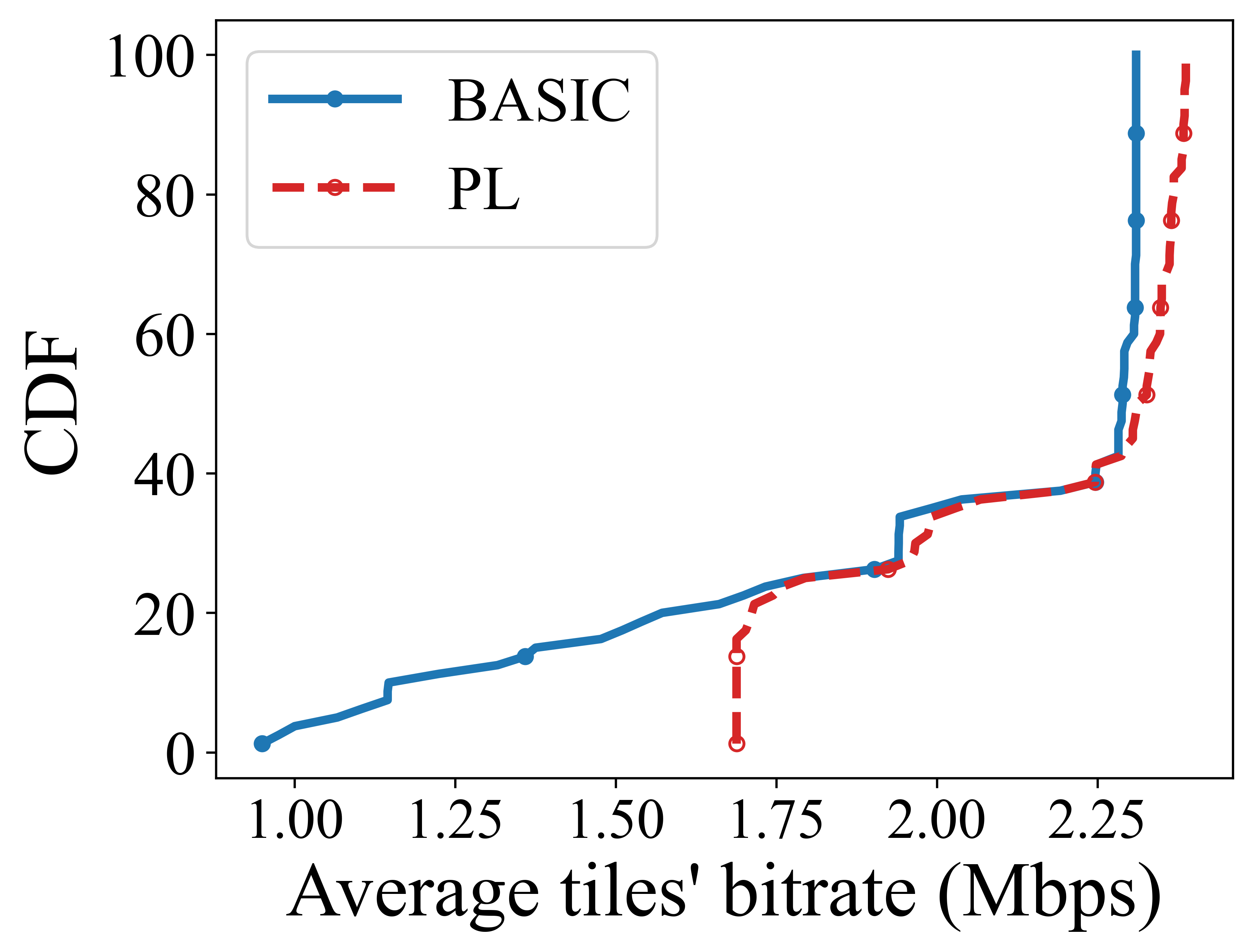}}\hspace{1mm}
	\subfigure{\includegraphics[width=0.32\textwidth]{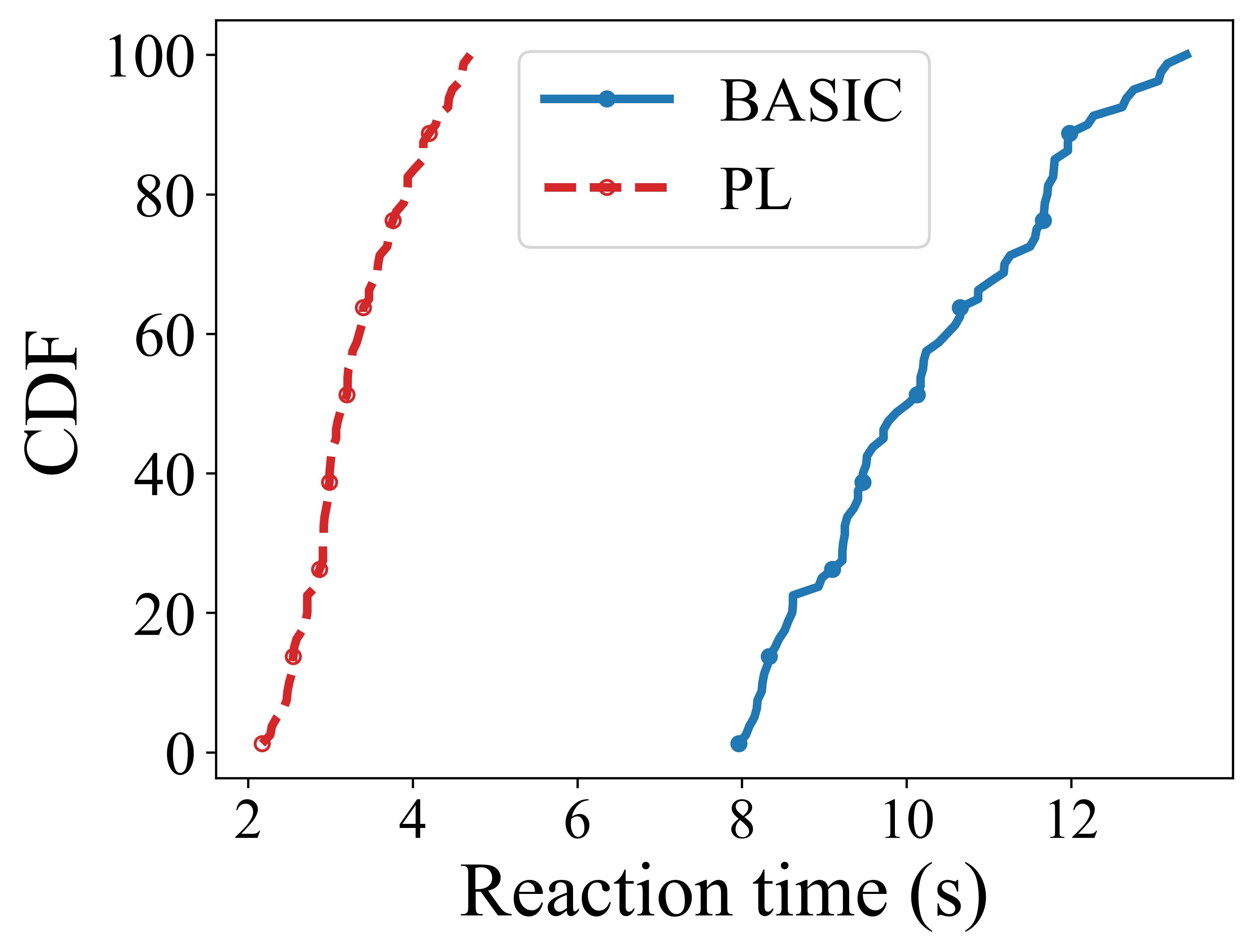}}\hspace{1mm}
	\subfigure{\includegraphics[width=0.32\textwidth]{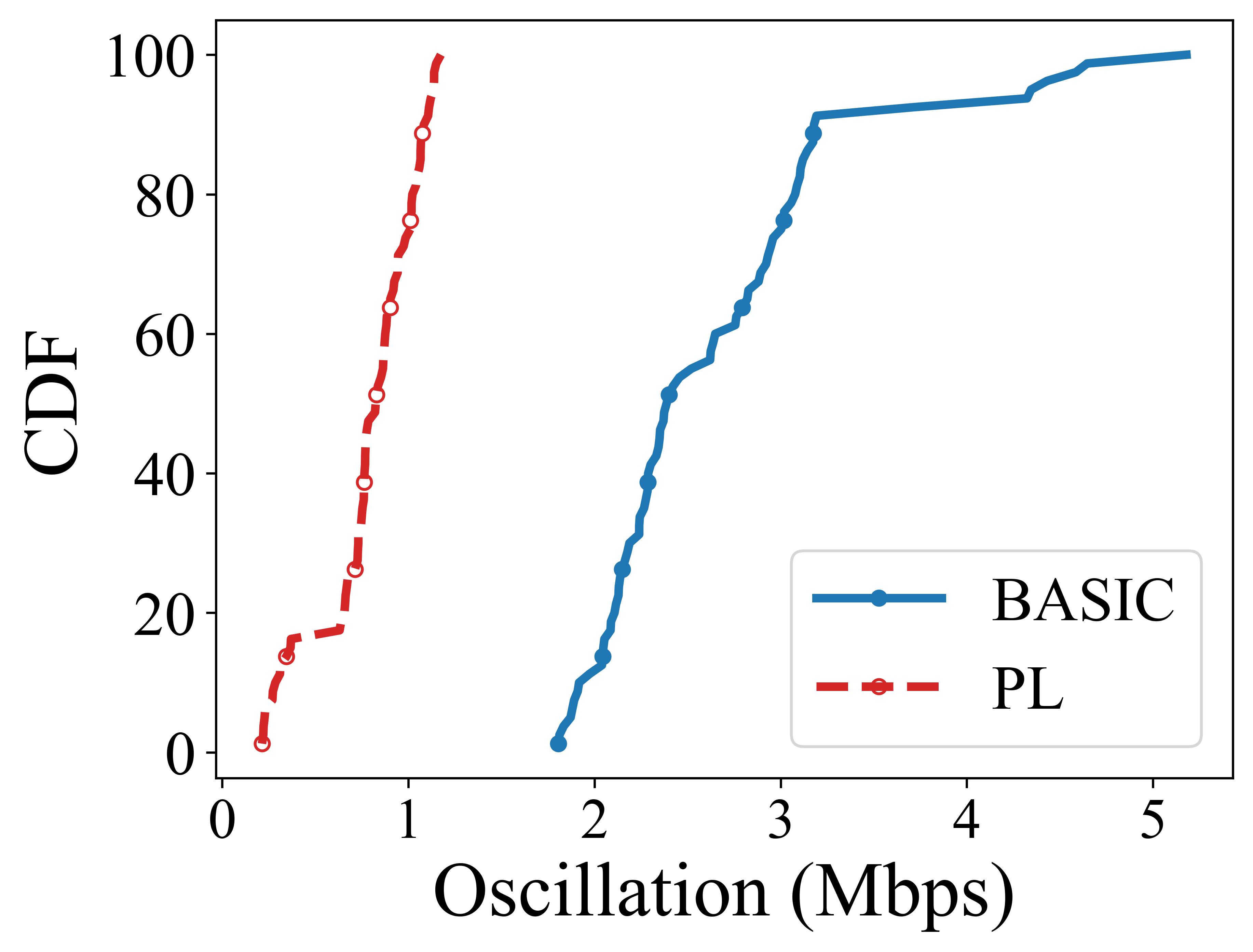}}
	\caption{The CDF of the average bitrate of any downloaded tile (left), reaction time (middle), and oscillation (right) of basic \algName and \hrPL using real network and head movement traces. \hrPL reduces the oscillation and reaction time by $70.9\%$ and $67.8\%$ respectively. }
	\label{fig:heuristic_set1}

\end{figure*}

\begin{figure*}[t]
	\centering
    
	\subfigure{\includegraphics[width=0.32\textwidth]{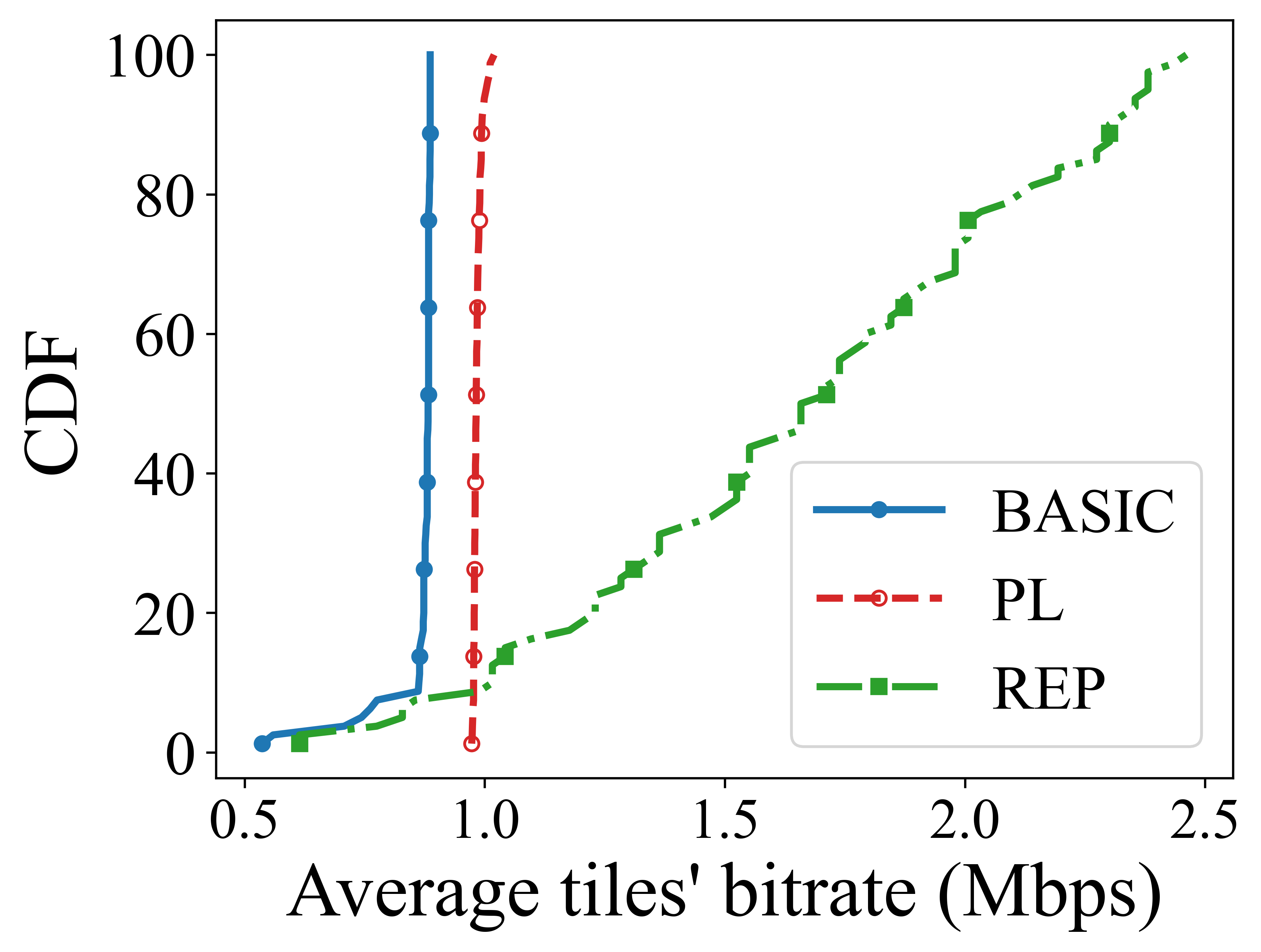}}\hspace{1mm}
	\subfigure{\includegraphics[width=0.32\textwidth]{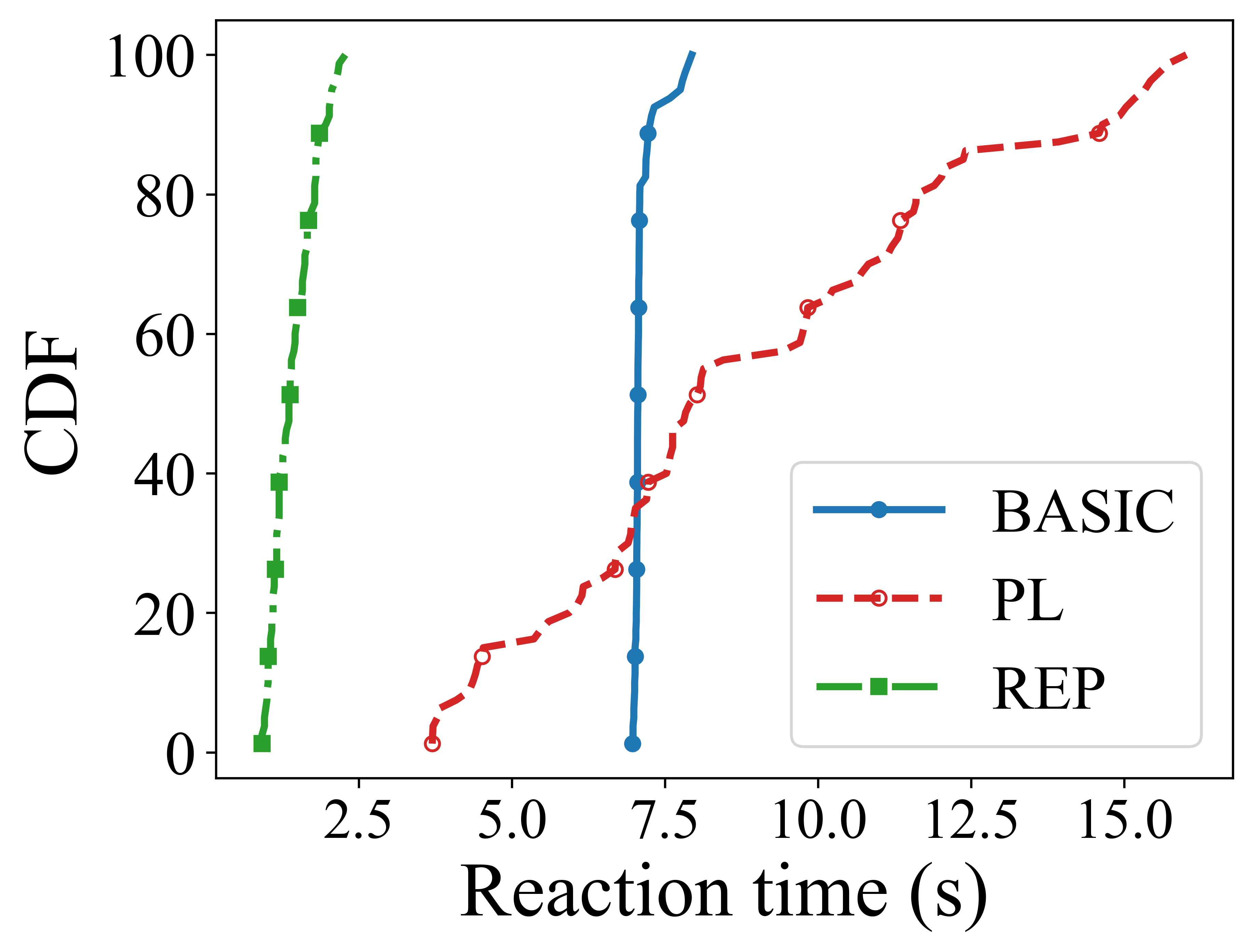}}\hspace{1mm}
	\subfigure{\includegraphics[width=0.32\textwidth]{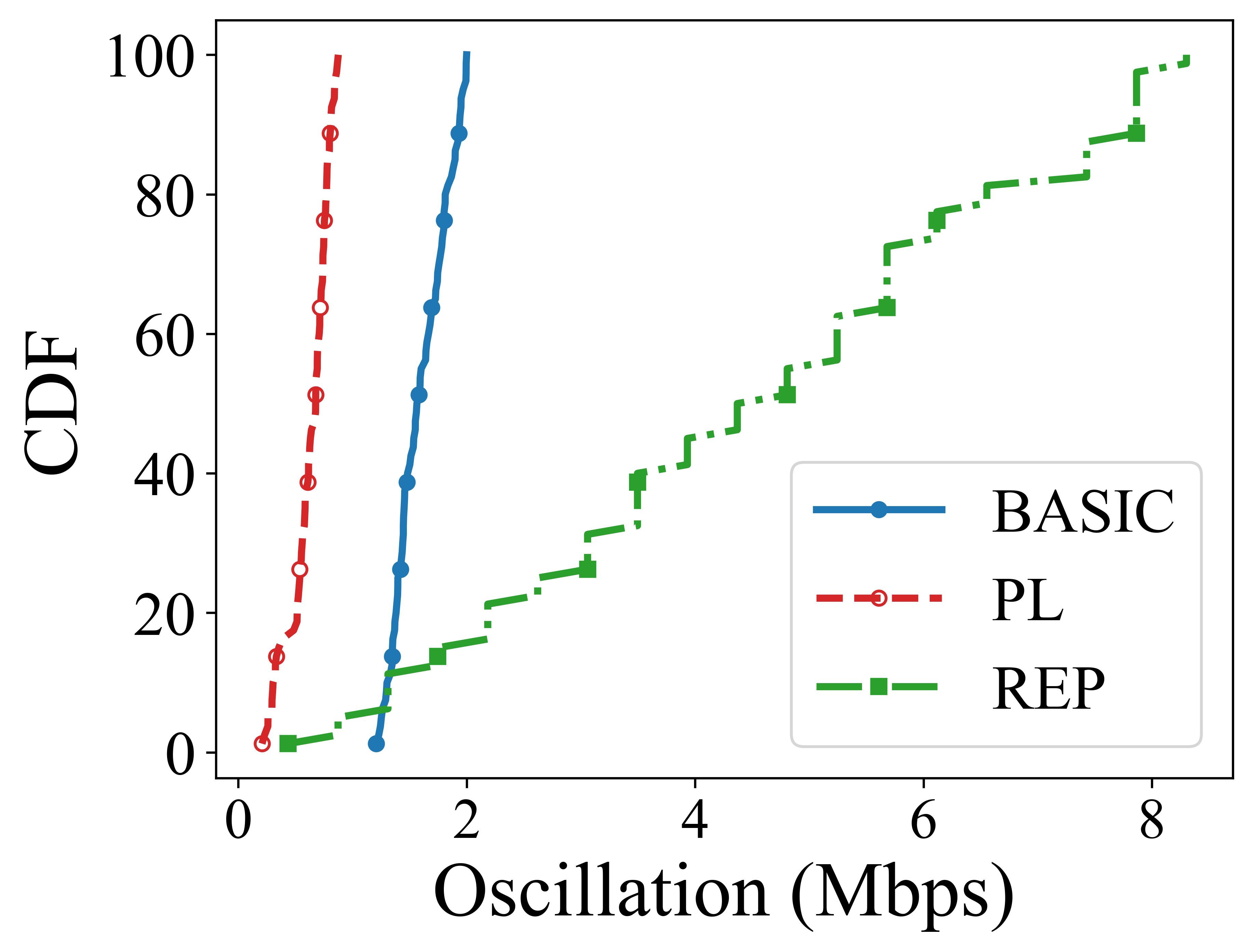}}
	\caption{The CDF of average bitrate of downloaded tiles (left), reaction time (middle), and oscillation (right) of basic \algName, \hrPL, and \hrREP using real network and head movement traces while the prediction of the head position dynamically got updated.  \hrREP improves the average bitrate of downloaded tiles up to $91.2\%$ compared to \algName BASIC, and reduces the reaction time by $80.0\%$. 
 }
	\label{fig:heuristic_dynamicProbs}
\end{figure*}

Figure~\ref{fig:heuristic_set1} shows the CDF plots of the average tiles' bitrate (left), the reaction time (middle), and oscillation (right, the average difference between the bitrate of two consecutive chunk for each tile) of 100 trials for accurate head position probability predictions. The results show that \hrPL significantly reduces the oscillation and reaction time of \algName. Since the \hrPL improves the bitrate of tiles during start and seek time, and these tiles are a low fraction of the entire video, the average bitrate of tiles that \hrPL prepared for the user is slightly better than the average bitrate of tiles \algName downloads. 

In Figure~\ref{fig:heuristic_dynamicProbs}, we report the result of the evaluation of \algName and heuristic versions against the noisy prediction of head positions. Specifically, we report the CDF plot of the average tiles' bitrate, reaction time, and the oscillation of \algName, \hrPL, and \hrREP. The results show that the average bitrate of \algName and \hrPL reduced compared to the case where accurate head position probabilities were available. On the other hand, \hrREP improves the average bitrate of \algName up to near $97.6\%$ and reduces the reaction time of \algName by $80.0\%$. Although \hrREP  could improve the average bitrate and the reaction time, it increases the oscillation. The average oscillation time for \algName was $1.6$ seconds, while this value for \hrREP  was $4.5$ seconds. Meanwhile, all two heuristic versions could keep the rebuffering as low as the rebuffering of \algName.

\textbf{Key takeaway.} Each extension of \algName improves the performance in certain aspects, such as bitrate or reaction time. However, each version has drawbacks that may result in lower performance in other aspects. Therefore, no version outperforms the others in all aspects, and depending on the application and user requirements, different versions may be suitable.

\section{Related Work}
\label{sec:relwork}

The prior literature extensively addresses the problem of bitrate and view adaptation in \threesixty video streaming. Previous works commonly employ various machine learning techniques  to predict user head movements and network throughput, incorporating these predictions into ABR algorithms~\cite{guo2024long, qian2018flare, li2023spherical, zhu2023toward, wang2024synergistic, zhang2023tile,setayesh2023content,yang2024360spred,wan2024predicting,yaqoob2023collaborative}. For example, \cite{qian2018flare} proposes a prediction-based approach and designs an ABR algorithm using historical data from \threesixty video streaming sessions. The focus of their work is on head movement prediction, while the ABR algorithm itself is a heuristic approach lacking rigorous optimization-based mechanisms.


Authors in~\cite{zhang2023ram360} propose a Lyapunov-based model to solve the \tABR problem, also utilized in~\cite{shen2019qoe}. They employ Lyapunov optimization for selecting and adjusting the bitrate of tiles, resulting in a nearly optimal ABR algorithm achieved through iterative updates to the tiles of a chunk. However, despite its near-optimality, this technique may suffer from a long reaction time and high wasted bandwidth due to iterative bitrate adjustment.
In another work,~\cite{jiang2019hierarchical} proposes a different approach by constructing a two-layered hierarchical buffer-based algorithm with short and long buffer layers. The prediction of \AF is used to perform short-term improvement. The long buffer layer tries to download the new tiles that are not available in the short buffer layer and will be played later. In another work,~\cite{park2019advancing} predicts the head movement by using a saliency map, tile probability heat map, and LSTM models and gives \tABR algorithm based on that.

In another category of work~\cite{zhang2019drl360,park2021adaptive,fu2021360hrl,kan2021rapt360, kim2020neural, wang2024madrl}, deep RL-based algorithms are developed for solving \tABR. They also use a dataset of the user's head position to train the model and find the optimal bitrate selection according to the predicted \AF. 

In a recent study~\cite{guo2024adaptive}, the authors proposed a non-uniform coding and transmission method that divides the entire 360-degree field into four regions: an attention area, an out-of-field-of-view area, a peripheral area, and other areas. They then introduced an online ABR algorithm that dynamically selects appropriate bitrates for each region. This work was further extended by another study~\cite{guo2024long} that optimized the transmission strategy for 360-degree videos using LSTM networks.

In \cite{yuan2019spatial}, \AF prediction is used to select proper bitrates for tiles in a predicted \AF, with the accuracy of prediction impacting the final bitrate selection. Other works such as \cite{nguyen2018your, wang2022salientvr, xie2018cls, petrangeli2017http, xie2017360probdash} also focus on \AF prediction. The main idea is that users have similar region-of-interest when watching the same video. They divide the users into clusters such that users inside each cluster have similar region-of-interest in most videos. Then they give \AF prediction based on the cluster of a given user and the historical head direction traces of users in a predicted cluster. While these approaches help reduce bandwidth waste, they still require an ABR algorithm to select bitrates within the predicted region. In contrast, \algName is an online algorithm with rigorous performance guarantees, solving the \tABR problem optimally.
Guan et al. \cite{guan2019pano} employ Model Predictive Control (MPC) to select the aggregate bitrate for a chunk, allocating it among tiles to maintain quality within the limited bitrate. In another category of research~\cite{he2018rubiks, rudow2023tambur}, an optimized coding/encoding algorithm minimizes bandwidth usage for \threesixty videos, evaluated using real 4K and 8K videos from YouTube. Their experiments use a straightforward ABR algorithm resembling \PDash (Section~\ref{sec:exp}).

\section{Conclusion and Future Directions}
\label{sec:conc}
In this paper, we formulated an optimization problem to maximize users' QoE in \threesixty video streaming applications. We proposed \algName, an online algorithm that achieves a provably near-optimal solution by selecting a proper bitrate for each tile of a \threesixty video to maximize quality while minimizing rebuffering rate. Our experimental results demonstrate that \algName outperforms several alternative algorithms across various network and head movement profiles. In future work, we aim to develop a data-driven and robust version of \algName that explicitly uses future predictions in decision-making while maintaining the algorithm's theoretical performance guarantees.


\section*{Acknowledgments}
This research was supported in part by NSF grants CAREER 2045641, CPS-2136199, CNS-2106299, CNS-2102963, CSR-1763617, CNS-2106463, and CNS-1901137. We acknowledge their financial assistance in making this project possible.

\bibliographystyle{unsrt}  
 \bibliography{references.bib}

\begin{thebibliography}{10}

\bibitem{youtube360}
Youtube.
\newblock youtube360.
\newblock \url{https://www.youtube.com/360}, 2022.
\newblock Accessed: 2022-03.

\bibitem{psvr}
Sony.
\newblock Sony playstaion vr.
\newblock \url{https://www.playstation.com/en-us/ps-vr2/}, 2022.
\newblock Accessed: 2022-03.

\bibitem{gARVR}
Google LLC.
\newblock Google ar/vr.
\newblock \url{https://arvr.google.com/ar/}, 2022.
\newblock Accessed: 2022-03.

\bibitem{cicso_vni}
CISCO.
\newblock Cisco mobile visual networking index (vni) forecast projects 7-fold
  increase in global mobile data traffic from 2016-2021.
\newblock \url{https://tinyurl.com/CICSO-netwok}, 2017.
\newblock Accessed: 2022-08.

\bibitem{ZinkSN19}
Michael Zink, Ramesh Sitaraman, and Klara Nahrstedt.
\newblock Scalable 360° video stream delivery: Challenges, solutions, and
  opportunities.
\newblock {\em Proceedings of the IEEE}, 107(4):639--650, 2019.

\bibitem{nguyen2018predictive}
Thanh~Cong Nguyen and Ji-Hoon Yun.
\newblock Predictive tile selection for 360-degree vr video streaming in
  bandwidth-limited networks.
\newblock {\em IEEE Communications Letters}, 22(9):1858--1861, 2018.

\bibitem{mangiante2017vr}
Simone Mangiante, Guenter Klas, Amit Navon, Zhuang GuanHua, Ju~Ran, and
  Marco~Dias Silva.
\newblock Vr is on the edge: How to deliver 360 videos in mobile networks.
\newblock In {\em Proceedings of the Workshop on Virtual Reality and Augmented
  Reality Network}, pages 30--35, 2017.

\bibitem{dasari2020streaming}
Mallesham Dasari, Arani Bhattacharya, Santiago Vargas, Pranjal Sahu, Aruna
  Balasubramanian, and Samir~R Das.
\newblock Streaming 360-degree videos using super-resolution.
\newblock In {\em IEEE INFOCOM}, pages 1977--1986. IEEE, 2020.

\bibitem{usnetworkAkamai2017}
Akamai.
\newblock Akamai’s [state of the internet].
\newblock \url{https://tinyurl.com/Akmai-internet-connectivity}, 2017.
\newblock Accessed: 2022-07.

\bibitem{usnetwork2021}
EtiSoftware.
\newblock Internet speed and subscriber dissatisfaction.
\newblock \url{https://tinyurl.com/network-speed}, 2021.
\newblock Accessed: 2022-07.

\bibitem{kua2017survey}
Jonathan Kua, Grenville Armitage, and Philip Branch.
\newblock {A Survey of Rate Adaptation Techniques for Dynamic Adaptive
  Streaming Over HTTP}.
\newblock {\em IEEE Communications Surveys \& Tutorials}, 19(3):1842--1866,
  2017.

\bibitem{han2020vivo}
Bo~Han, Yu~Liu, and Feng Qian.
\newblock {ViVo: Visibility-aware mobile volumetric video streaming}.
\newblock In {\em Proceedings of the 26th MobiCom}, pages 1--13, 2020.

\bibitem{mao2017neural}
Hongzi Mao, Ravi Netravali, and Mohammad Alizadeh.
\newblock Neural adaptive video streaming with pensieve.
\newblock In {\em Proceedings of ACM SIGCOMM}, pages 197--210, 2017.

\bibitem{yin2015control}
Xiaoqi Yin, Abhishek Jindal, Vyas Sekar, and Bruno Sinopoli.
\newblock {A control-theoretic approach for dynamic adaptive video streaming
  over HTTP}.
\newblock In {\em Proceedings of ACM SIGCOMM}, pages 325--338, 2015.

\bibitem{zhang2021sensei}
Xu~Zhang, Yiyang Ou, Siddhartha Sen, and Junchen Jiang.
\newblock {SENSEI: Aligning Video Streaming Quality with Dynamic User
  Sensitivity}.
\newblock In {\em NSDI}, pages 303--320, 2021.

\bibitem{spiteri2020bola}
Kevin Spiteri, Rahul Urgaonkar, and Ramesh~K Sitaraman.
\newblock {BOLA: Near-optimal bitrate adaptation for online videos}.
\newblock {\em IEEE/ACM Transactions on Networking}, 28(4):1698--1711, 2020.

\bibitem{kim2020neural}
Jaehong Kim, Youngmok Jung, Hyunho Yeo, Juncheol Ye, and Dongsu Han.
\newblock Neural-enhanced live streaming: Improving live video ingest via
  online learning.
\newblock In {\em Proceedings of ACM SIGCOMM}, pages 107--125, 2020.

\bibitem{ghabashneh2023dragonfly}
Ehab Ghabashneh, Chandan Bothra, Ramesh Govindan, Antonio Ortega, and Sanjay
  Rao.
\newblock Dragonfly: Higher perceptual quality for continuous 360 video
  playback.
\newblock In {\em Proceedings of the ACM SIGCOMM 2023 Conference}, pages
  516--532, 2023.

\bibitem{li2019very}
Chenge Li, Weixi Zhang, Yong Liu, and Yao Wang.
\newblock Very long term field of view prediction for 360-degree video
  streaming.
\newblock In {\em IEEE MIPR}, pages 297--302. IEEE, 2019.

\bibitem{fan2017fixation}
Ching-Ling Fan, Jean Lee, Wen-Chih Lo, Chun-Ying Huang, Kuan-Ta Chen, and
  Cheng-Hsin Hsu.
\newblock Fixation prediction for 360 video streaming in head-mounted virtual
  reality.
\newblock In {\em Proceedings of the 27th NOSSDAV}, pages 67--72, 2017.

\bibitem{xu2018gaze}
Yanyu Xu, Yanbing Dong, Junru Wu, Zhengzhong Sun, Zhiru Shi, Jingyi Yu, and
  Shenghua Gao.
\newblock Gaze prediction in dynamic 360 immersive videos.
\newblock In {\em proceedings of the IEEE Conference on Computer Vision and
  Pattern Recognition}, pages 5333--5342, 2018.

\bibitem{xu2018predicting}
Mai Xu, Yuhang Song, Jianyi Wang, MingLang Qiao, Liangyu Huo, and Zulin Wang.
\newblock Predicting head movement in panoramic video: A deep reinforcement
  learning approach.
\newblock {\em IEEE transactions on pattern analysis and machine intelligence},
  41(11):2693--2708, 2018.

\bibitem{ban2018cub360}
Yixuan Ban, Lan Xie, Zhimin Xu, Xinggong Zhang, Zongming Guo, and Yue Wang.
\newblock {Cub360: Exploiting cross-users behaviors for viewport prediction in
  360 video adaptive streaming}.
\newblock In {\em IEEE ICME}, pages 1--6. IEEE, 2018.

\bibitem{xu2018tile}
Zhimin Xu, Yixuan Ban, Kai Zhang, Lan Xie, Xinggong Zhang, Zongming Guo,
  Shengbin Meng, and Yue Wang.
\newblock {Tile-based QoE-driven HTTP/2 streaming system for 360 video}.
\newblock In {\em ICMEW}, pages 1--4. IEEE, 2018.

\bibitem{feng2020livedeep}
Xianglong Feng, Yao Liu, and Sheng Wei.
\newblock Livedeep: Online viewport prediction for live virtual reality
  streaming using lifelong deep learning.
\newblock In {\em IEEE Conference on Virtual Reality and 3D User Interfaces
  (VR)}, pages 800--808. IEEE, 2020.

\bibitem{feng2021liveroi}
Xianglong Feng, Weitian Li, and Sheng Wei.
\newblock Liveroi: region of interest analysis for viewport prediction in live
  mobile virtual reality streaming.
\newblock In {\em Proceedings of the 12th ACM MMSys}, pages 132--145, 2021.

\bibitem{zeynali2024bola360}
Ali Zeynali, Mohammad~H Hajiesmaili, and Ramesh~K Sitaraman.
\newblock {BOLA360: Near-optimal View and Bitrate Adaptation for 360-degree
  Video Streaming}.
\newblock In {\em Proceedings of the 15th ACM Multimedia Systems Conference},
  pages 12--22, 2024.

\bibitem{ozcinar2017viewport}
Cagri Ozcinar, Ana De~Abreu, and Aljosa Smolic.
\newblock Viewport-aware adaptive 360 video streaming using tiles for virtual
  reality.
\newblock In {\em IEEE ICIP}, pages 2174--2178. IEEE, 2017.

\bibitem{xie2017360probdash}
Lan Xie, Zhimin Xu, Yixuan Ban, Xinggong Zhang, and Zongming Guo.
\newblock {360probdash: Improving qoe of 360 video streaming using tile-based
  HTTP adaptive streaming}.
\newblock In {\em Proceedings of the 25th ACM MM}, pages 315--323, 2017.

\bibitem{wang2022salientvr}
Shibo Wang, Shusen Yang, Hailiang Li, Xiaodan Zhang, Chen Zhou, Chenren Xu,
  Feng Qian, Nanbin Wang, and Zongben Xu.
\newblock {SalientVR: saliency-driven mobile 360-degree video streaming with
  gaze information}.
\newblock In {\em Proceedings of MobiCom}, pages 542--555, 2022.

\bibitem{qian2018flare}
Feng Qian, Bo~Han, Qingyang Xiao, and Vijay Gopalakrishnan.
\newblock Flare: Practical viewport-adaptive 360-degree video streaming for
  mobile devices.
\newblock In {\em Proceedings of ACM MobiCom}, pages 99--114, 2018.

\bibitem{guan2019pano}
Yu~Guan, Chengyuan Zheng, Xinggong Zhang, Zongming Guo, and Junchen Jiang.
\newblock {Pano: Optimizing 360 video streaming with a better understanding of
  quality perception}.
\newblock In {\em Proceedings of the ACM SIGCOMM}, pages 394--407. 2019.

\bibitem{park2021mosaic}
Sohee Park, Arani Bhattacharya, Zhibo Yang, Samir~R Das, and Dimitris Samaras.
\newblock Mosaic: Advancing user quality of experience in 360-degree video
  streaming with machine learning.
\newblock {\em IEEE Transactions on Network and Service Management},
  18(1):1000--1015, 2021.

\bibitem{spiteri2019theory}
Kevin Spiteri, Ramesh Sitaraman, and Daniel Sparacio.
\newblock From theory to practice: Improving bitrate adaptation in the dash
  reference player.
\newblock {\em ACM Transactions on Multimedia Computing, Communications, and
  Applications}, 15(2s):1--29, 2019.

\bibitem{ratcliff2020thinvr}
Joshua Ratcliff, Alexey Supikov, Santiago Alfaro, and Ronald Azuma.
\newblock {ThinVR: Heterogeneous microlens arrays for compact, 180 degree FOV
  VR near-eye displays}.
\newblock {\em IEEE transactions on visualization and computer graphics},
  26(5):1981--1990, 2020.

\bibitem{bao2016shooting}
Yanan Bao, Huasen Wu, Tianxiao Zhang, Albara~Ah Ramli, and Xin Liu.
\newblock Shooting a moving target: Motion-prediction-based transmission for
  360-degree videos.
\newblock In {\em 2016 IEEE International Conference on Big Data (Big Data)},
  pages 1161--1170. IEEE, 2016.

\bibitem{chen2020sparkle}
Jinyu Chen, Xianzhuo Luo, Miao Hu, Di~Wu, and Yipeng Zhou.
\newblock Sparkle: User-aware viewport prediction in 360-degree video
  streaming.
\newblock {\em IEEE Transactions on Multimedia}, 23:3853--3866, 2020.

\bibitem{park2019navigation}
Jounsup Park and Klara Nahrstedt.
\newblock Navigation graph for tiled media streaming.
\newblock In {\em Proceedings of the 27th ACM MM}, pages 447--455, 2019.

\bibitem{qian2016optimizing}
Feng Qian, Lusheng Ji, Bo~Han, and Vijay Gopalakrishnan.
\newblock Optimizing 360 video delivery over cellular networks.
\newblock In {\em Proceedings of the 5th Workshop on All Things Cellular:
  Operations, Applications and Challenges}, pages 1--6, 2016.

\bibitem{xie2018cls}
Lan Xie, Xinggong Zhang, and Zongming Guo.
\newblock {Cls: A cross-user learning based system for improving qoe in
  360-degree video adaptive streaming}.
\newblock In {\em Proceedings of the 26th ACM MM}, pages 564--572, 2018.

\bibitem{neely2010stochastic}
Michael~J Neely.
\newblock Stochastic network optimization with application to communication and
  queueing systems.
\newblock {\em Synthesis Lectures on Communication Networks}, 3(1):1--211,
  2010.

\bibitem{neely2012dynamic}
Michael~J Neely.
\newblock Dynamic optimization and learning for renewal systems.
\newblock {\em IEEE Transactions on Automatic Control}, 58(1):32--46, 2012.

\bibitem{reichl2013logarithmic}
Peter Reichl, Bruno Tuffin, and Raimund Schatz.
\newblock Logarithmic laws in service quality perception: where microeconomics
  meets psychophysics and quality of experience.
\newblock {\em Telecommunication Systems}, 52(2):587--600, 2013.

\bibitem{hu2019optimization}
Han Hu, Cheng Zhan, Jianping An, and Yonggang Wen.
\newblock {Optimization for HTTP adaptive video streaming in UAV-enabled
  relaying system}.
\newblock In {\em ACM ICC}, pages 1--6. IEEE, 2019.

\bibitem{huang2014buffer}
Te-Yuan Huang, Ramesh Johari, Nick McKeown, Matthew Trunnell, and Mark Watson.
\newblock A buffer-based approach to rate adaptation: Evidence from a large
  video streaming service.
\newblock In {\em Proceedings of the 2014 ACM conference on SIGCOMM}, pages
  187--198, 2014.

\bibitem{hao2021buffer}
Mingyue Hao, Jinghao Yuan, Bingcong Lu, Li~Song, Rong Xie, and Wenjun Zhang.
\newblock {Buffer Displacement Based Online Learning Algorithm For Low Latency
  HTTP Adaptive Streaming}.
\newblock In {\em 2021 IEEE International Symposium on Broadband Multimedia
  Systems and Broadcasting (BMSB)}. IEEE, 2021.

\bibitem{bokani2016prehensive}
A~Bokani, M~Hassan, SS~Kanhere, J~Yao, and G~Zhong.
\newblock Comprehensive mobile bandwidth traces from vehicular networks.
\newblock In {\em Proceedings of the 7th ACM MMSys}, pages 1--6, 2016.

\bibitem{IDLAB}
Ghent University.
\newblock {4G/LTE Bandwidth Logs}.
\newblock \url{https://users.ugent.be/~jvdrhoof/dataset-4g/}, 2019.
\newblock Accessed: 2022-06.

\bibitem{vanderHooft2016}
J.~van~der Hooft, S.~Petrangeli, T.~Wauters, R.~Huysegems, P.~R. Alface,
  T.~Bostoen, and F.~De~Turck.
\newblock {HTTP/2-Based Adaptive Streaming of HEVC Video Over 4G/LTE Networks}.
\newblock {\em IEEE Communications Letters}, 20(11):2177--2180, 2016.

\bibitem{chrmDev}
Google LLC.
\newblock Google-chrome: Chrome devtools protocol.
\newblock
  \url{https://chromedevtools.github.io/devtools-protocol/tot/Network/}, 2023.
\newblock Accessed: 2023-01.

\bibitem{wu2017dataset}
Chenglei Wu, Zhihao Tan, Zhi Wang, and Shiqiang Yang.
\newblock {A dataset for exploring user behaviors in VR spherical video
  streaming}.
\newblock In {\em MMSys}, pages 193--198, 2017.

\bibitem{guo2024long}
Jia Guo, Chengrui Li, Jinqi Zhu, Xiang Li, Qian Gao, Yunhe Chen, and Weijia
  Feng.
\newblock {Long Short-Term Memory-Based Non-Uniform Coding Transmission
  Strategy for a 360-Degree Video}.
\newblock {\em Electronics}, 13(16):3281, 2024.

\bibitem{li2023spherical}
Jie Li, Ling Han, Chong Zhang, Qiyue Li, and Zhi Liu.
\newblock {Spherical convolution empowered viewport prediction in 360 video
  multicast with limited FoV feedback}.
\newblock {\em ACM Transactions on Multimedia Computing, Communications and
  Applications}, 19(1):1--23, 2023.

\bibitem{zhu2023toward}
Yucheng Zhu, Xiongkuo Min, Dandan Zhu, Guangtao Zhai, Xiaokang Yang, Wenjun
  Zhang, Ke~Gu, and Jiantao Zhou.
\newblock {Toward visual behavior and attention understanding for augmented 360
  degree videos}.
\newblock {\em ACM Transactions on Multimedia Computing, Communications and
  Applications}, 19(2s):1--24, 2023.

\bibitem{wang2024synergistic}
Yumei Wang, Junjie Li, Zhijun Li, Simou Shang, and Yu~Liu.
\newblock {Synergistic Temporal-Spatial User-Aware Viewport Prediction for
  Optimal Adaptive 360-Degree Video Streaming}.
\newblock {\em IEEE Transactions on Broadcasting}, 2024.

\bibitem{zhang2023tile}
Zhiahao Zhang, Yiwei Chen, Weizhan Zhang, Caixia Yan, Qinghua Zheng, Qi~Wang,
  and Wangdu Chen.
\newblock {Tile classification based viewport prediction with multi-modal
  fusion transformer}.
\newblock In {\em Proceedings of the 31st ACM International Conference on
  Multimedia}, pages 3560--3568, 2023.

\bibitem{setayesh2023content}
Mehdi Setayesh and Vincent~WS Wong.
\newblock {A Content-based Viewport Prediction Framework for 360° Video Using
  Personalized Federated Learning and Fusion Techniques}.
\newblock In {\em 2023 IEEE International Conference on Multimedia and Expo
  (ICME)}, pages 654--659. IEEE, 2023.

\bibitem{yang2024360spred}
Qin Yang, Wenxuan Gao, Chenglin Li, Hao Wang, Wenrui Dai, Junni Zou, Hongkai
  Xiong, and Pascal Frossard.
\newblock {360Spred: Saliency Prediction for 360-Degree Videos Based on 3D
  Separable Graph Convolutional Networks}.
\newblock {\em IEEE Transactions on Circuits and Systems for Video Technology},
  2024.

\bibitem{wan2024predicting}
Zhaolin Wan, Han Qin, Ruiqin Xiong, Zhiyang Li, Xiaopeng Fan, and Debin Zhao.
\newblock {Predicting 360° Video Saliency: A ConvLSTM Encoder-Decoder Network
  with Spatio-temporal Consistency}.
\newblock {\em IEEE Journal on Emerging and Selected Topics in Circuits and
  Systems}, 2024.

\bibitem{yaqoob2023collaborative}
Abid Yaqoob and Gabriel-Miro Muntean.
\newblock {A Collaborative Trajectory-Oriented Viewport Prediction for
  on-Demand and Live 360 VR Video Streaming}.
\newblock In {\em 2023 IEEE 20th International Conference on Smart Communities:
  Improving Quality of Life using AI, Robotics and IoT (HONET)}, pages 1--6.
  IEEE, 2023.

\bibitem{zhang2023ram360}
Haodan Zhang, Yixuan Ban, Zongming Guo, Ken Chen, and Xinggong Zhang.
\newblock {RAM360: Robust Adaptive Multi-layer 360 Video Streaming with
  Lyapunov Optimization}.
\newblock {\em IEEE Transactions on Multimedia}, 2023.

\bibitem{shen2019qoe}
Wang Shen, Lianghui Ding, Guangtao Zhai, Ying Cui, and Zhiyong Gao.
\newblock {A QoE-oriented saliency-aware approach for 360-degree video
  transmission}.
\newblock In {\em IEEE VCIP}, pages 1--4. IEEE, 2019.

\bibitem{jiang2019hierarchical}
Zhiqian Jiang, Xu~Zhang, Wei Huang, Hao Chen, Yiling Xu, Jenq-Neng Hwang, Zhan
  Ma, and Jun Sun.
\newblock A hierarchical buffer management approach to rate adaptation for
  360-degree video streaming.
\newblock {\em IEEE Transactions on Vehicular Technology}, 69(2):2157--2170,
  2019.

\bibitem{park2019advancing}
Sohee Park, Arani Bhattacharya, Zhibo Yang, Mallesham Dasari, Samir~R Das, and
  Dimitris Samaras.
\newblock Advancing user quality of experience in 360-degree video streaming.
\newblock In {\em IFIP Networking}. IEEE, 2019.

\bibitem{zhang2019drl360}
Yuanxing Zhang, Pengyu Zhao, Kaigui Bian, Yunxin Liu, Lingyang Song, and
  Xiaoming Li.
\newblock {DRL360: 360-degree video streaming with deep reinforcement
  learning}.
\newblock In {\em IEEE INFOCOM}, pages 1252--1260. IEEE, 2019.

\bibitem{park2021adaptive}
Sohee Park, Minh Hoai, Arani Bhattacharya, and Samir~R Das.
\newblock Adaptive streaming of 360-degree videos with reinforcement learning.
\newblock In {\em Proceedings of the IEEE/CVF Winter Conference on Applications
  of Computer Vision}, pages 1839--1848, 2021.

\bibitem{fu2021360hrl}
Jun Fu, Chen Hou, and Zhibo Chen.
\newblock {360hrl: Hierarchical Reinforcement Learning Based Rate Adaptation
  for 360-Degree Video Streaming}.
\newblock In {\em IEEE VCIP}, pages 1--5. IEEE, 2021.

\bibitem{kan2021rapt360}
Nuowen Kan, Junni Zou, Chenglin Li, Wenrui Dai, and Hongkai Xiong.
\newblock {RAPT360: Reinforcement learning-based rate adaptation for 360-degree
  video streaming with adaptive prediction and tiling}.
\newblock {\em IEEE Transactions on Circuits and Systems for Video Technology},
  32(3):1607--1623, 2021.

\bibitem{wang2024madrl}
Haopeng Wang, Zijian Long, Haiwei Dong, and Abdulmotaleb El~Saddik.
\newblock {MADRL-Based Rate Adaptation for 360 Video Streaming With
  Multi-Viewpoint Prediction}.
\newblock {\em IEEE Internet of Things Journal}, 2024.

\bibitem{guo2024adaptive}
Jia Guo, Shiqiang Li, Jinqi Zhu, Xiang Li, Bowen Sun, and Weijia Feng.
\newblock {Adaptive Transmission Strategy for Non-Uniform Coding of 360
  Videos.}
\newblock {\em Electronics (2079-9292)}, 13(16), 2024.

\bibitem{yuan2019spatial}
Hui Yuan, Shiyun Zhao, Junhui Hou, Xuekai Wei, and Sam Kwong.
\newblock Spatial and temporal consistency-aware dynamic adaptive streaming for
  360-degree videos.
\newblock {\em IEEE Journal of Selected Topics in Signal Processing},
  14(1):177--193, 2019.

\bibitem{nguyen2018your}
Anh Nguyen, Zhisheng Yan, and Klara Nahrstedt.
\newblock Your attention is unique: Detecting 360-degree video saliency in
  head-mounted display for head movement prediction.
\newblock In {\em Proceedings of the 26th ACM MM}, pages 1190--1198, 2018.

\bibitem{petrangeli2017http}
Stefano Petrangeli, Viswanathan Swaminathan, Mohammad Hosseini, and Filip
  De~Turck.
\newblock {An HTTP/2-based adaptive streaming framework for 360 virtual reality
  videos}.
\newblock In {\em Proceedings of ACM MM}, pages 306--314, 2017.

\bibitem{he2018rubiks}
Jian He, Mubashir~Adnan Qureshi, Lili Qiu, Jin Li, Feng Li, and Lei Han.
\newblock {Rubiks: Practical 360-degree streaming for smartphones}.
\newblock In {\em Proceedings of the 16th ACM MobiSys}, pages 482--494, 2018.

\bibitem{rudow2023tambur}
Michael Rudow, Francis~Y Yan, Abhishek Kumar, Ganesh Ananthanarayanan, Martin
  Ellis, and KV~Rashmi.
\newblock Tambur: Efficient loss recovery for videoconferencing via streaming
  codes.
\newblock In {\em NSDI}, pages 953--971, 2023.

\end{thebibliography}

\Urlmuskip=0mu plus 1mu\relax



\end{document}